\RequirePackage{fix-cm}
\documentclass[smallcondensed]{svjour3}     

\makeatletter
\def\makeheadbox{{%
\hbox to0pt{\vbox{\baselineskip=10dd\hrule\hbox
to\hsize{\vrule\kern3pt\vbox{\kern3pt
\hbox{\bfseries\@journalname\ 2015}
\hbox{DOI 10.1007/s00236-015-0221-6 }
\kern3pt}\hfil\kern3pt\vrule}\hrule}%
\hss}}}
\makeatother

\smartqed  
\usepackage{graphicx}
\usepackage{xspace}
\usepackage{hyperref,breakurl}
\usepackage{amsmath,amsfonts,amssymb}
\DeclareSymbolFont{frenchscript}{OMS}{ztmcm}{m}{n}
\DeclareMathSymbol{\A}{\mathord}{frenchscript}{65}   
\DeclareMathSymbol{\B}{\mathord}{frenchscript}{66}   
\DeclareMathSymbol{\FS}{\mathord}{frenchscript}{70}  
\DeclareMathSymbol{\HC}{\mathord}{frenchscript}{72} 
\DeclareMathSymbol{\K}{\mathord}{frenchscript}{75}   
\DeclareMathSymbol{\OA}{\mathord}{frenchscript}{79}  
\DeclareMathSymbol{\pow}{\mathord}{frenchscript}{80} 
\newcommand{\denote}[1]{[\hspace{-1.4pt}[#1]\hspace{-1.4pt}]}  
\newcommand{\NN}{
    \ensuremath{%
        \mathop{\rm I\mkern-2.5mu N}%
        \nolimits%
    }%
}
\newcommand{\T}{{\rm T}}                             
\newcommand{\SC}{{\rm G}}                            
\newcommand{\E}{P}                                   
\newcommand{\F}{Q}                                   
\newcommand{\G}{\cal G}
\newcommand{\AI}{A}                                  
\makeatletter
\def\moverlay{\mathpalette\mov@rlay}
\def\mov@rlay#1#2{\leavevmode\vtop{%
   \baselineskip\z@skip \lineskiplimit-\maxdimen
   \ialign{\hfil$\m@th#1##$\hfil\cr#2\crcr}}}
\newcommand{\charfusion}[3][\mathord]{
    #1{\ifx#1\mathop\vphantom{#2}\fi
        \mathpalette\mov@rlay{#2\cr#3}
      }
    \ifx#1\mathop\expandafter\displaylimits\fi}
\makeatother
\newcommand{\dcup}{\charfusion[\mathbin]{\cup}{\mbox{\Large$\cdot$}}}

\newcommand{\plat}[1]{\raisebox{0pt}[0pt][0pt]{#1}}  
\def\precond#1{{\vphantom{#1}}^\bullet #1}
\def\postcond#1{{#1}^\bullet}

\makeatletter
\spn@wtheorem{observation}{Observation}{\bfseries}{\rmfamily}
\makeatother
\providecommand{\urlalt}[2]{\href{#1}{#2}}
\providecommand{\doi}[1]{doi:\urlalt{http://dx.doi.org/#1}{#1}}
\makeatletter
\def\comesfrom{\@transition\leftarrowfill}
\def\goesto{\@transition\rightarrowfill}
\def\ngoesto{\@transition\nrightarrowfill}
\def\Goesto{\@transition\Rightarrowfill}
\def\nGoesto{\@transition\nRightarrowfill}
\def\xmapsto{\@transition\mapstofill}
\def\nxmapsto{\@transition\nmapstofill}
\def\@transition#1{\@@transition{#1}}
\newbox\@transbox
\newbox\@arrowbox
\newbox\@downbox
\def\@@transition#1#2%
   {\setbox\@transbox\hbox
      {\vrule height 1.5ex depth .9ex width 0ex\hskip0.25em$\scriptstyle#2$\hskip0.25em}
   \ifdim\wd\@transbox<1.5em
      \setbox\@transbox\hbox to 1.5em{\hfil\box\@transbox\hfil}\fi
   \setbox\@arrowbox\hbox to \wd\@transbox{#1}
   \ht\@arrowbox\z@\dp\@arrowbox\z@
   \setbox\@transbox\hbox{$\mathop{\box\@arrowbox}\limits^{\box\@transbox}$}
   \dp\@transbox\z@\ht\@transbox 10pt
   \mathrel{\box\@transbox}}
\def\nrightarrowfill{$\m@th\mathord-\mkern-6mu%
  \cleaders\hbox{$\mkern-2mu\mathord-\mkern-2mu$}\hfill
  \mkern-6mu\mathord\not\mkern-2mu\mathord\rightarrow$}
\def\Rightarrowfill{$\m@th\mathord=\mkern-6mu%
  \cleaders\hbox{$\mkern-2mu\mathord=\mkern-2mu$}\hfill
  \mkern-6mu\mathord\Rightarrow$}
\def\nRightarrowfill{$\m@th\mathord=\mkern-6mu%
  \cleaders\hbox{$\mkern-2mu\mathord=\mkern-2mu$}\hfill
  \mkern-6mu\mathord\not\mathord\Rightarrow$}
\def\mapstofill{$\m@th\mathord\mapstochar\mathord-\mkern-6mu%
  \cleaders\hbox{$\mkern-2mu\mathord-\mkern-2mu$}\hfill
  \mkern-6mu\mathord\rightarrow$}
\def\nmapstofill{$\m@th\mathord\mapstochar\mathord-\mkern-6mu%
  \cleaders\hbox{$\mkern-2mu\mathord-\mkern-2mu$}\hfill
  \mkern-6mu\mathord\not\mkern-2mu\mathord\rightarrow$}
\makeatother 
\newcommand{\ar}[1]{\mathrel{\goesto{#1}}}            

\newcommand{\Thm}[1]{Theorem~\ref{thm:#1}}
\newcommand{\Cor}[1]{Corollary~\ref{cor:#1}}

\newcommand{\Lem}[1]{Lemma~\ref{lem:#1}}
\newcommand{\Def}[1]{Definition~\ref{df:#1}}
\newcommand{\Ex}[1]{Example~\ref{ex:#1}}

\newcommand{\Part}[1]{Part {\ref{part:#1}}}
\newcommand{\Sect}[1]{Section~\ref{sec:#1}}
\newcommand{\SSect}[1]{Section~\ref{ssec:#1}}

\newcommand{\Fig}[1]{Figure~\ref{fig:#1}}
\newcommand{\Tab}[1]{Table~\ref{tab:#1}}

\newcommand{\myref}[1]{\hyperlink{lab:#1}{\sc (#1)}\xspace}
\newcommand{\skipped}[1]{}
\renewcommand{\inf}{^{\omega}}
\newcommand{\mand}{\&\xspace}

\makeatletter
\newcounter {part}
\renewcommand\thepart{\@Roman\c@part}
\def\part#1#2{%
\vspace{6ex}
\noindent {\Large\bf\boldmath Part \refstepcounter{part}\thepart\ \ \ #1\label{part:#2}}

\vspace{-3.5ex}
}

\makeatother

\newcounter{Hequation}

\makeatletter
\g@addto@macro\equation{\stepcounter{Hequation}}
\makeatother

\journalname{Acta Informatica}
\begin{document}

\title{CCS: It's not Fair!%
\thanks{%
    NICTA is funded by the Australian Government through the
    Department of Communications and the Australian Research Council
    through the ICT Centre of Excellence Program.
}
}
\subtitle{Fair Schedulers cannot be implemented in CCS-like languages even under progress and certain fairness assumptions}

\author{
    Rob van Glabbeek \and 
    Peter H\"ofner
}

\institute{
    R.\,J.\ van Glabbeek \at
	    NICTA and UNSW.
        \email{rvg@cs.stanford.edu} \and
    P.\ H\"ofner \at
	    NICTA and UNSW.
	\email{Peter.Hoefner@nicta.com.au}
}

\date{~}

\maketitle

\begin{quote}\it
It is our great pleasure to dedicate this paper to Walter Vogler on the occasion of his 60th birthday.
We have combined two of Walter's main interests: Petri nets and process algebra. 
In fact, we proved a result about Petri nets that had been proven before by Walter, but in a restricted form, 
as we discovered only after finishing our proof. 
We also transfer this result to the process algebra CCS.
Beyond foundational research in the theory of concurrent systems, Walter achieved excellent results in related subjects such 
as temporal logic and efficiency.
In addition to being a dedicated researcher, he is also 
meticulous in all of his endeavours, including his writing.
As a consequence his scientific papers tend to contain no flaws,
which is just one of the reasons that makes reading them so enjoyable.

It's fair to say: ``CCS Walter!''---Congratulations and Continuous Success!

\end{quote}

\begin{abstract}
In the process algebra community it is sometimes suggested that,
on some level of abstraction, any distributed system can be
modelled in standard process-algebraic specification formalisms like CCS\@.
This sentiment is strengthened by results testifying that CCS,
like many similar formalisms, is Turing powerful and provides a mechanism for 
interaction.
This paper counters that sentiment by presenting a simple fair
scheduler---one that in suitable variations occurs in many
distributed systems---of which no implementation can be
expressed in CCS, unless CCS is enriched with a fairness assumption.

Since  Dekker's and Peterson's mutual exclusion protocols implement fair schedulers,
it follows that these protocols cannot be rendered correctly in CCS
without imposing a fairness assumption. Peterson expressed this algorithm correctly in pseudocode
without resorting to a fairness assumption, so it furthermore follows that  CCS lacks the expressive
power to accurately capture such pseudocode.
\end{abstract}

\part{Motivation \& Discussion}{one}
\section{Background}

In the process algebra community it is
generally taken for granted
that, on some level of abstraction, any distributed system can be
modelled in standard process-algebraic specification formalisms like CCS~\cite{Mi89}.

Of course, if a distributed system has features related to time, probability, broadcast communication
or anything else that is not innately modelled in CCS, yet essential to adequately describe the
distributed system under consideration, appropriate extensions are needed,
such as timed process algebras (e.g., \cite{ReedRoscoe86,HennessyRegan95,BB96,LV01,CVJ02}),
probabilistic process algebras (e.g., \cite{HJ90})
or calculi for broadcast communication (e.g., \cite{CBS}).
This paper is not concerned with such features.

The relative expressiveness of process algebras is a well-studied subject \cite{Va93,Parrow08,Gorla10a},
 and in this area CCS-like process algebras are considered far from universally expressive.
In \cite{vG12} for instance it is pointed out that the parallel composition operator of CSP \cite{BHR84,Ho85} cannot
be expressed in CCS\@.
The priority operator of \cite{BBK87b}
is a good example of an
operator that cannot be expressed in any of the standard process algebras such as CCS, CSP, ACP~\cite{BK86} or LOTOS \cite{BB87}.
These results focus, however, on the possibility of expressing \emph{operators}---composing a process out of
one or more components---as CCS-contexts; they cast no doubt on the possibility of expressing actual
processes in CCS\@.

Beside operators, it has also be shown that there exist examples of process \emph{specifications} that cannot be faithfully rendered in
CCS-like formalisms (cf.~\cite{vG05d}). We will illustrate this in \SSect{spec vs implementation}. In
this paper we distinguish process specifications from actual processes that could in principle be implemented and executed.
Again, the evidence presented casts no doubt on the possibility of expressing actual processes in CCS\@.

Incorporating these clarifications of our meaning, we expect that many concurrency experts feel that,
up to an adequate level of abstraction, any reactive system can be rendered in CCS\@.
This sentiment is strengthened by results testifying that CCS,
like many similar formalisms, is Turing powerful~\cite{Mi89}.
As a manifestation of this, any computable partial function $f:\Sigma^* \rightarrow \Sigma^*$ over some
finite alphabet $\Sigma$ can be modelled by a CCS context $P[\,\_\!\_\hspace{2.3pt}]$, such that,
for any input word $w=a_1a_2\dots a_n\in\Sigma^*$, encoded as a CCS expression
$W:=a_1.a_2.\dots.a_n.z.0$ featuring an end-of-input marker $z$, the process $P[W]$ computes
forever without performing any visible actions if $f(w)$ is undefined, and otherwise performs the
sequence of visible actions $f(w)$, followed by an end-of-output marker $z$.

It is sometimes argued that Turing machines are an inadequate formalism to capture interactive
behaviour as displayed by today's computers~\cite{Wegner97,LW01}.
The main
argument is that Turing machines are function-based and 
calculate, for a given finite input, one output;
this paradigm does not do justice to the ongoing interactions between a reactive system and its
environment.
To add ongoing interactivity to Turing machines, \emph{interaction machines} are proposed in
\cite{Wegner97}, and formalised in \cite{GSAS04} as \emph{persistent Turing machines}. Likewise,
\cite{BLT11} proposes \emph{reactive Turing machines}.
Since standard process algebras like CCS are already equipped with interaction primitives,
they can surely also model computations on persistent or reactive Turing machines.
All this strengthens the feeling
that standard process algebras, such as CCS,
are powerful enough to specify any distributed system.

\section{Fairness Assumptions}{\label{sec:fairness}}

Before presenting evidence that CCS and related formalisms cannot correctly specify every distributed system,
some explanation is in order on our understanding of `correctly'. This is best illustrated by an example.
\advance\textheight 1pt

Consider the CCS agent identifier $E$ with defining equation \plat{$E \stackrel{{\it def}}{=} a.E + b.0$}.
The question is whether this is a good rendering of a process that is guaranteed to eventually
perform the action $b$. The answer depends on whether we incorporate a fairness assumption in the
semantics of CCS\@. A \emph{strong fairness} assumption requires that if a task (here $b$) is enabled
infinitely often, but allowing interruptions during which it is not enabled, it will eventually be
scheduled \cite{GPSS80,LPS81}.
Making such an assumption allows us to infer that indeed the process $E$ will eventually do a $b$.%
\footnote{In \cite{BBK87a} a form of reasoning using a particularly strong global fairness assumption
was integrated in the axiomatic framework of ACP, and shown to be compatible with the notion of
weak bisimulation commonly taken as the semantic basis for CCS\@.}

\advance\textheight -1pt
It depends on the context of the application of CCS whether it is appropriate to make such fairness assumptions.
For the verification of the alternating bit protocol, for instance, fairness assumptions are
indispensable \cite{BergstraKlop85}. But in some situations they allow us to reach conclusions that are
patently false. In the example above for instance, let $a$ be an unsuccessful attempt to dial a
number or an unreliable mobile phone, and $b$ a successful one. The system $E$ simply
retries after each unsuccessful attempt. Whether it ever succeeds in performing $b$ depends very
much on how unreliable the phone is. If there is a fixed positive probability on success,
the strong fairness assumption appears warranted. Yet, if the phone is completely dead, it is not,
and the conclusion that we eventually succeed in dialling is false.
In fact, when assuming strong fairness we loose the expressiveness to describe by a finite
recursive specification like $E$ a system such as the above interaction with the unreliable 
telephone that \emph{does} allow an infinite run with only $a$s.

As evidence that not every distributed system can be rendered correctly in CCS, we will describe a
\emph{fair scheduler} as a counterexample. Yet, our fair scheduler can be rendered in CCS very
easily, if only we are willing to postulate a fairness property to support its correctness.
However, considering the above example and the fact that we 
may reach wrong conclusions, this is a price we are not willing to pay.

Our fair scheduler is not merely an `artificial' CCS specification; it is implemented in many working distributed systems, and (unlike the alternating bit
protocol) its correctness should not be contingent on any fairness assumption whatsoever.
This is another reason why we do not want to invoke fairness to achieve a correct rendering in CCS\@.

Yet, we do find it reasonable to equip CCS with two assumptions that are weaker than strong fairness,
namely \emph{progress} and \emph{justness}. A progress assumption is what is needed to infer that
the CCS process $b.0$ will eventually do a $b$, and a justness assumption allows us to infer that
the parallel composition $A|b.0$ with \plat{$A \stackrel{{\it def}}{=} a.A$} will eventually do a $b$.
If our task is to specify in CCS a process $\cal{B}$ that will eventually do a $b$,
then, when assuming strong fairness,
the processes~$E$,  $A|b.0$ and $b.0$ are acceptable implementations of $\cal{B}$.
If we assume justness, but not fairness, this selection shrinks to $A|b.0$ and $b.0$, and if we only
assume progress, we have to give up on $A|b.0$ as well. When not even assuming progress, $\cal{B}$
cannot be rendered in CCS at all.
Assuming progress and justness only, $A|b.0$ models a process that will
eventually do a $b$, whereas $E$ can be used to characterise the above mentioned interaction with the unreliable telephone,
which allows an infinite sequence of $a$s only.

So, when we claim that a fair scheduler cannot be implemented in CCS, we mean that it
cannot be implemented in CCS+justness, CCS+progress or CCS without any progress assumption.
It can be implemented in CCS+strong fairness, however. 

To pinpoint the borderline, 
consider a \emph{weak fairness} or \emph{justice} assumption~\cite{GPSS80,LPS81}. This assumption
requires that if a task, from some point onwards, is
perpetually enabled, it will eventually be scheduled. What this means depends on our interpretation
of `perpetual'. If `perpetual' simply means `in each state', then a weak fairness assumption
is all that is needed to assure that the process $E$ will eventually do a $b$.\footnote{The process
  $E'$ with {\plat{$E' \stackrel{{\it def}}{=} a.a.E' + b.0$}} on the other hand really needs
  strong fairness.} A weak fairness assumption in this
sense is enough to correctly render a fair scheduler in CCS\@. If, on the other hand, the execution
of the $a$-transition of $E$ counts as a (short) interruption of the enabledness of $b$, then
justice can be shown to coincide with justness \cite{GH14}; as we will show,
this is not enough to render a fair scheduler in CCS\@.

\section{Specifications versus Actual Processes}\label{ssec:spec vs implementation}

Consider the system specification $\G$ expressed in\vspace{2pt}
CCS as $A|B$ with \plat{$A \stackrel{{\it def}}{=} a.A$} and \plat{$B \stackrel{{\it def}}{=} b.B$},
but with the added requirement that all infinite executions should have infinitely many occurrences
of $a$ as well as $b$. Here $a$ and $b$ could be seen as two tasks that need to be scheduled again
and again. The left-hand component $A$ of the parallel composition tries to perform task $a$
infinitely often, and the right-hand component tries to perform task $b$ infinitely often.
The process $A|B$ by itself, as specified in CCS,\vspace{2pt} is normally deemed equivalent to the process $C$,
defined by \plat{$C \stackrel{{\it def}}{=} a.C + b.C$}, and---in the absence of a justness or
fairness assumption---offers no guarantee that a single $b$
will ever happen. It could be that, due to unfortunate scheduling, at each time a choice is made, task $a$
is chosen. The challenge in specifying the fair version $\G$ of this process in CCS is 
how to ensure that sooner or later a $b$ will happen, without simply invoking a fairness or justness
assumption, and without setting any fixed limit on the number of $a$s that can happen beforehand.

Accordingly, solutions have appeared in the literature that change the operational semantics of CCS
in such a way that $A|B$ will surely do a $b$ eventually. In~\cite{Plotkin82} for instance,
parallel operators $\|_m$ are used that, each time a $b$ occurs, nondeterministically select a
number $m$ and guarantee that from that point onwards at most $m$ occurrences of $a$ happen
before the next $b$. Another solution along these lines is proposed in \cite{CS84}, whereas
\cite{CBV06} solves the problem by harvesting the power added by the treatment of time in the timed
process algebra PAFAS \cite{CVJ02}.

In relation to the above challenge it would be trivial to specify \emph{some} process that makes
sure that tasks $a$ and $b$ are each scheduled infinitely often; a particularly simple way to
achieve this is through the CCS specification $D$,\vspace{2pt} given by
\plat{$D \stackrel{{\it def}}{=} a.b.D$}; that is, to alternate each of the two tasks.
This is a \emph{round-robin} solution.
It could be seen as a particular implementation of $\G$.
The reason that such a solution is not chosen is that it fails to capture the full generality of
the original specification, in which arbitrary many $a$s may come between any two $b$s.

Any real-life implementation of $\G$ on a physical computer is unlikely to capture the full
generality of its specification, but rather goes a few steps towards the round-robin solution.
For this reason, one could argue that $\G$ does not constitute an example of a distributed system
that cannot be rendered in CCS, but rather one of a \emph{specification} that
cannot be rendered in CCS\@. As such, it falls out of the scope of this paper.

\section{Our Contributions}\label{ssec:contribution}

This paper counters the sentiment that CCS-like process algebras are powerful
enough to represent arbitrary distributed systems by presenting a particularly simple system of which no
implementation can be expressed in CCS. 
The reason we use CCS is that it is among the most well-known standard process algebras, while having a fairly easy to explain
syntax and semantics. However, we believe the same result, with
essentially the same proof, could be obtained for COSY \cite{LTS79}, CSP \cite{BHR84,Ho85},
ACP \cite{BK86}, LOTOS \cite{BB87}, $\mu$CRL \cite{GP95}, the $\pi$-calculus \cite{MPW92}, etc.

Our system is a \emph{fair scheduler}. It can receive two
kinds of requests $r_1$ and $r_2$ from its environment on separate channels, and is guaranteed to perform a
task---\emph{granting} the request---in response. Our fairness requirement rules out a scheduler that
may fail to ever grant a request of type $r_1$ because it is consistently busy granting requests $r_2$.

Such schedulers occur
(in suitable variations)
in many distributed systems.
Examples are \emph{First in First out}%
\footnote{Also known as First Come First Served (FCFS)},
\emph{Round Robin}, and
\emph{Fair Queueing} 
scheduling algorithms\footnote{\url{http://en.wikipedia.org/wiki/Scheduling_(computing)}}
as used in network routers~\cite{rfc970,Nagle87} and operating systems~\cite{Kleinrock64},
or the \emph{Completely Fair Scheduler},\footnote{\url{http://en.wikipedia.org/wiki/Completely_Fair_Scheduler}}
which is the default scheduler of the Linux kernel since version 2.6.23.

If $\cal F$ stands for the most general specification of our scheduler, our claim entails that 
$\cal F$ cannot be rendered in CCS\@. However, accurately expressing $\cal F$ in CCS can
be seen as a luxury problem. Here we would accept \emph{any} implementation of $\cal F$, under the
broadest definition of implementation that makes sense for this problem---a round-robin solution for
instance would be totally acceptable---and what we show is that even that is
impossible.

As is common, we employ a version of CCS that allows the use of arbitrary sets of recursive equations to define processes.
As is trivial to show, \emph{any} labelled transition system, computable or 
not, can be modelled up to strong bisimulation equivalence as an expression in this
language. Hence, our result implies that no implementation of the fair scheduler $\cal F$ can be modelled as
a labelled transition system modulo strong bisimulation equivalence.

In this paper we will use a semantics of CCS incorporating a justness assumption. It distinguishes the strongly bisimilar systems $A|B$
and $C$ mentioned above, on grounds that $A|B$ can be understood to always perform infinitely many
$a$s as well as $b$s, whereas $C$ might perform an infinite sequence of $b$s while discarding the
$a$-option all the time. This semantics increases the power of CCS in specifying fair schedulers,
and thereby strengthens our result that no implementation of the fair scheduler $\cal F$ can be
expressed. It thereby becomes stronger than the result that no implementation of $\cal F$ can be
rendered as a labelled transition system modulo strong bisimulation equivalence.

To prove our result, we show that our fair scheduler cannot be expressed in terms of  safe
Petri nets. The result for CCS then follows by reduction: an adequate Petri net semantics of CCS shows that
if the scheduler could be expressed in CCS, it could also be expressed as a Petri net.

The reason we resort to Petri nets to prove our main theorem is that Petri nets offer a structural
characterisation of what it means for a transition to be continuously enabled in a run of the
represented system from some state onwards. This is exploited in the proofs of
Lemmas~\ref{lem:one step embellishment}--\ref{lem:extension to complete path}.
It would be much harder to prove their counterparts directly in terms of the
labelled transition system of CCS\@.

In different formulations, our impossibility result for Petri nets was established earlier by
Vogler in \cite{Vogler02} and by Kindler \mand Walter in \cite{KW97}, but in both cases side
conditions were imposed that inhibit lifting these results to CCS\@.
The proof of \cite[Lemma 6.1]{Vogler02} considers only finite Petri nets.
The argument would extend to finitely branching nets, but not to all Petri nets that arise as the
semantics of CCS expressions. The proof of \cite{KW97} is restricted to Petri nets that interact
with their environment through an interface of a particular shape, and it is not a priori clear that
this does not cause a loss of generality. However, in \Sect{characterisation} we study a similar interface in
the context of CCS and show that it does not limit generality.

Although our fair scheduler cannot be expressed in a standard process algebra like CCS,
we believe there are many extensions in the literature in which this can be done quite easily.
In \Sect{fair spec} for instance, we specify it in a formalism that could be called CCS+LTL.
The use of a priority operator appears to be sufficient as well.\vspace{-6pt}

\section{Peterson's and Dekker's Mutual Exclusion Protocols}

Since Peterson's and Dekker's mutual exclusion protocols yield instances of our fair scheduler,
it follows that these protocol cannot be rendered correctly in CCS
without imposing a fairness assumption. Nevertheless, implementations of these algorithms in CCS
or similar formalisms occur frequently in the literature, and almost never a fairness assumption is invoked.
Moreover, for each of these two protocols, its various renderings differ only in insignificant details.
Our result implies that these common renderings cannot be correct.
Usually, only safety properties of these protocols are shown: never are two processes simultaneously in
the critical section. The problem is with the liveness property: any process that is ready to enter
the critical section will eventually do so. We found four papers that claim to establish essentially this
property, of which only one invokes a fairness assumption.
We will indicate in which way the other three do not establish the right liveness property.

Peterson expressed his protocol correctly in pseudocode
without resorting to a fairness assumption, although progress and justness are assumed implicitly.
It follows that Peterson's pseudocode does not admit an accurate translation into CCS\@.
We pinpoint the problem in this paper.

\section{Overview}

In \Part{one} we discussed (informally) the goal we set out to achieve, and 
why we believe it is important and surprising at the same time.

\Part{two} formalises our results, while providing explanations of the choices made in this formalisation.
In particular, \Sect{fair scheduler} presents an informal description of our fair
scheduler $\cal F$.  \Sect{CCS} presents CCS\@.
\Sect{progress} makes a 
progress assumption on the semantics of CCS and argues that it is
useful to set apart a set of non-blocking actions.
\Sect{justness} formalises the justness assumption discussed above and presents a semantics of CCS in which a process
$P$ is modelled as a state in a labelled transition system together with a set of (possibly
infinite) paths in that transition system starting from $P$ that model its valid runs.
 \Sect{fair spec} gives a formal specification of $\cal F$.
Since we aim to show that no implementation of $\cal F$ can be
specified in CCS, the specification of $\cal F$ cannot be given in CCS
either.  Instead we specify $\cal F$ as a CCS expression augmented
with a \emph{fairness specification}.  This follows the traditional
approach of TLA~\cite{TLA} and other formalisms \cite{Fr86}, ``in which
first the legal computations are specified, and then a fairness notion
is used to exclude some computations which otherwise would be legal'' \cite{AFK88}.
In \Sect{no fair scheduler} we state our main result, saying that no fair scheduler---that is: no
implementation of $\cal F$---can be expressed in CCS\@.
\Sect{characterisation} reformulates this result, so that it is
independent of the concept of an action being perpetually enabled in a run of the represented system.
In \Sect{discussion} we conclude that mutual exclusion
protocols, like the algorithms from Dekker or Peterson, cannot be
rendered correctly in CCS without imposing a fairness condition.
We also investigate the apparent contradiction with the fact that
several research papers claim to achieve exactly this.
We end this section with a result by Corradini, Di Berardini \mand Vogler,
showing where a fairness assumption is needed for a rendering
of Dekker's protocol in a process algebra to be correct.

\Part{three} deals with proving our main result.
In \Sect{nets} we formulate our claim that no fair scheduler can be modelled as a 
safe Petri net. This claim is proven in \Sect{proof PN}.
In \Sect{operational PN} an operational Petri net semantics of CCS is
presented, following the work of Degano, De Nicola \mand Montanari.
From this, the proof of our main result is obtained in \Sect{proof conclusion}.
A few concluding remarks are made in \Sect{conclusion}.

\part{Formalisation}{two}

\section{A Fair Scheduler}\label{sec:fair scheduler}

Our fair scheduler is a reactive system with two input channels: one
on which it can receive requests $r_1$ from its environment and
one on which it can receive requests $r_2$. We allow the scheduler to
be too busy shortly after receiving a request $r_i$ to accept another
request $r_i$ on the same channel. However, the system will always
return to a state where it remains ready to accept the next request $r_i$
until $r_i$ arrives. In case no request arrives it 
remains ready forever. The environment is under no
obligation to issue requests, or to ever stop issuing requests.
Hence for any numbers $n_1$ and $n_2\in\NN\cup\{\infty\}$ there is at least one run of the system in
which exactly that many requests of type $r_1$ and $r_2$ are received.

Every request $r_i$ asks for a task $t_i$ to be executed.
The crucial property of the fair scheduler is that it will eventually
grant any such request. Thus, we require that in any run of the system
each occurrence of $r_i$ will be followed by an occurrence of $t_i$.
In Linear-time Temporal Logic (LTL)~\cite{Pnueli77} this can be stated as
\[
\mathbf{G}(r_i \Rightarrow \mathbf{F}(t_i))\,,\quad i\in\{1,2\}\,.
\]
\noindent
Note that it may happen that the environment issues request $r_1$
three times in a row before the scheduler got a change to schedule
task $t_1$. In that case the scheduler may fulfil its obligation by
scheduling task $t_1$ just once. Hence it need not keep a counter of
outstanding requests.\footnote{This relaxed requirement only
  serves to increase the range of acceptable schedulers, thereby
  strengthening our impossibility result. It by no means rules out a
  scheduler that schedules task $t_1$ exactly once for each request
  $r_1$ received.}

We are not interested in implementations of the scheduler that just
schedule both tasks infinitely often without even listening to the
requests. Hence we require that in any partial run of the scheduler
there may not be more occurrences of $t_i$ than of $r_i$, for $i=1,2$.

The last requirement is that between each two occurrences of $t_i$ and
$t_j$ for $i,j\in\{1,2\}$ an intermittent activity $e$ is
scheduled.\footnote{Our specification places no restrictions on the
  presence or absence of any additional occurrences of $e$. This again
  increases the range of acceptable implementations.}
This requirement will rule out fully concurrent implementations, in
which there are parallel components for task $t_{1}$ and task $t_{2}$ that do not
interact in any way.

\section{The Calculus of Communicating Systems}\label{sec:CCS}

\begin{table}[t]
\normalsize
\begin{center}
\framebox{$\begin{array}{ccc}
\alpha.\E \ar{\alpha} \E &
&\displaystyle\frac{\E_j\ar{\alpha} \E'}{\sum_{i\in I}\E_i \ar{\alpha} \E'}~~(j\mathbin\in I)\\[4ex]
\displaystyle\frac{\E\ar{\alpha} \E'}{\E|\F \ar{\alpha} \E'|\F} &
\displaystyle\frac{\E\ar{a} \E' ,~ \F \ar{\bar{a}} \F'}{\E|\F \ar{\tau} \E'| \F'} &
\displaystyle\frac{\F \ar{\alpha} \F'}{\E|\F \ar{\alpha} \E|\F'}\\[4ex]
\displaystyle\frac{\E \ar{\alpha} \E'}{\E\backslash a \ar{\alpha}\E'\backslash a}~~
                     (a\mathbin{\neq}\alpha\mathbin{\neq}\bar{a}) &
\displaystyle\frac{\E \ar{\alpha} \E'}{\E[f] \ar{f(\alpha)} \E'[f]} &
\displaystyle\frac{\E \ar{\alpha} \E'}{A\ar{\alpha}\E'}~~(A \stackrel{{\it def}}{=} P)
\end{array}$}
\end{center}
\caption{Structural operational semantics of CCS}
\label{tab:CCS}
\end{table}

The Calculus of communicating systems (CCS) \cite{Mi89} is parametrised with sets ${\A}$ of {\em names} and $\K$ of {\em agent identifiers};
\vspace{2pt}%
each $\AI\in\K$ comes with a defining equation \plat{$\AI \stackrel{{\it def}}{=} P$}
with $P$ being a CCS expression as defined below.
The set $\bar{\A}$ of {\em co-names} is $\bar\A:=\{\bar{a}
\mid a\in {\A}\}$, and 
the set $\HC$ of \emph{handshake actions} is 
\plat{$\HC:=\A \dcup \bar\A$}, the disjoint union of the names and co-names.
The function $\bar{\rule{0pt}{7pt}.}$ is extended to $\HC$ by declaring
$\bar{\bar{\mbox{$a$}}}:=a$. Finally, \plat{$Act := \HC\dcup \{\tau\}$} is the set of
{\em actions}. Below, $a$, $b$, $c$ range over $\HC$,
$\alpha$, $\beta$ over $Act$ and $\AI$ over $\K$.
A \emph{relabelling} is a function $f\!:\HC\mathbin\rightarrow \HC$ satisfying
$f(\bar{a})=\overline{f(a)}$; it extends to $Act$ by $f(\tau):=\tau$.
The set $\T_{\rm CCS}$ of CCS expressions is the smallest set including:
\begin{center}
\begin{tabular}{@{}l@{~~}l@{\qquad\quad}l@{~~}l@{\qquad\quad}l@{~~}l@{}}
$\AI$ &  \emph{agent identifier} &
$\alpha.\E$  & \emph{prefixing}&
$\sum_{i\in I}\E_i$  & \emph{choice} \\
$\E|\F$ & \emph{parallel composition}&
$\E\backslash a$  & \emph{restriction} &
$\E[f]$ &  \emph{relabelling} \\
\end{tabular}
\end{center}
\noindent for $\E,\E_i,\F \in\T_{\rm CCS}$, index sets $I$, and relabellings $f$.
We write $\E_1+\E_2$ for $\sum_{i\in I}\E_i$ if $I=\{1,2\}$, and $0$ if $I=\emptyset$.
The semantics of CCS is given by the labelled transition relation
$\mathord\rightarrow \subseteq \T_{\rm CCS}\times Act \times\T_{\rm CCS}$, where the transitions 
\plat{$\E\ar{\alpha}\F$}
are derived from the rules of \Tab{CCS}.
The pair $\langle \T_{\rm CCS}, \rightarrow\rangle$ is called the \emph{labelled transition system} (LTS)
of CCS\@.

\section{The Necessity of Output Actions}\label{sec:progress}

\newcommand{\D}{\ensuremath{F}}
\newcommand{\Do}{\ensuremath{F_{0}}}
Before attempting to specify our scheduler in CCS, let us have a look at a
simpler problem: the same scheduler, but with only one type of request $r$,
and one type of task $t$ to be scheduled. A candidate CCS
specification of such a scheduler is the process $\Do$, defined by\vspace{-1ex}
\[ \Do \stackrel{{\it def}}{=} r.e.t.\Do\,.\vspace{-2pt}\]
As stated in Section~\ref{sec:fair scheduler}, the scheduler is called fair if 
every request $r$ is eventually followed by the requested task $t$; so we 
want to ensure the property
$\mathbf{G}(r \Rightarrow \mathbf{F}(t))$.\footnote{When assuming that this formula holds,
  $\Do$ trivially satisfies the other properties required in \Sect{fair scheduler}:
the system will always return to a state where it remains ready to
accept the next request $r$ until it arrives; in any partial run there are no more
occurrences of $t$ than of $r$, and between each two occurrences of
$t$ the action $e$ is scheduled.}
However, we cannot guarantee that this property actually holds for process $\Do$.
The reason is that the process might remain in the state $s$ 
reached by taking transition $e$ without ever performing the 
action $t$. In any formalism that allows to remain in a state
even when there are enabled actions, no useful liveness property about
processes can ever be guaranteed. One therefore often makes a \emph{progress assumption},
saying that the system will not idle as long as it can make progress.
Armed with this assumption, it appears fair to say that the process $\Do$ satisfies the required
property $\mathbf{G}(r \Rightarrow \mathbf{F}(t))$.

However, by symmetry, the same line of reasoning would allow us to derive that $\Do$ satisfies
$\mathbf{G}(t \Rightarrow \mathbf{F}(r))$, i.e.\ each execution of $t$ will be followed by a new
request. Yet, this is something we specifically do not want to assume: the action $r$ is meant to be
fully under the control of the environment, and it may very well happen that at some point the
environment stops making further requests. A particular instance of this is when the environment is modelled by
a CCS context such as $(\bar r.0 | \,\_\!\_\,)\backslash r$; in this context the process $\Do$ will
receive the request $r$ only once.

Hence, we reject the validity of $\mathbf{G}(t \Rightarrow \mathbf{F}(r))$ based on environments
such as $(\bar r.0 | \_\!\_\hspace{1pt})\backslash r$.
However, the same reasoning allows environments such as $(\bar r.\bar t.\bar r.0 | \_\!\_\hspace{1pt})\backslash r\backslash t$
that do not allow the task $t$ to be executed more than once. The existence of such environments
totally defeats our scheduler, or any other one.

Thus, for a fair scheduler to make sense, we need to consider environments that have full control
over the action $r$ but cannot sabotage the mission of our scheduler by disallowing
tasks $t$ and $e$. We formalise this by calling $t$ and $e$ \emph{output actions}.
An \emph{output action}~\cite[Section 9.1]{TR13} is
an activity of our system that cannot be stopped by its environment; or, equivalently, considering
an action $t$ to be \emph{output} means that we choose not to consider environments that can block $t$.
In our schedulers, $t$ and $e$ are output actions, whereas $r$ is not.

Let CCS$^{!}$ be the variant of CCS that is parametrised not only by sets ${\A}$ of names and
$\K$ of agent identifiers, but also by a set $\OA$ of output actions. The only further
difference with CCS of \Sect{CCS} is that \plat{$Act := \HC\dcup \OA\dcup \{\tau\}$}, and a relabelling $f$
extends to $Act$ by $f(\alpha) := \alpha$ for all $\alpha \in \OA \dcup\{\tau\}$.
CCS$^!$ can be seen as an extension of CCS with output actions, but it can just as well be seen as a
restriction of CCS in which for some of the names there are no co-names and no restriction operators.

It should be noted that CCS already features the concept of an \emph{internal} action, namely
$\tau$, of which it is normally assumed that it cannot be blocked by the environment.
Yet, for the purpose of specifying our scheduler, the r\^ole of the output action $t$ cannot be played
by $\tau$, for the internal action is supposed to be unobservable and is easily
abstracted away. Output actions share the feature of internal actions that whether they
occur or not is determined by the internal work of the specified system only; yet at the same
time they are observable by the environment in which the system is running.\footnote{The output
    and internal actions of CCS$^!$ are similar to the output and internal action of I/O automata \cite{LT89}.
    However, the remaining actions of I/O automata are \emph{input actions} that are totally under
    the control of the environment of the modelled system. In CCS, on the other hand, the default type
    of action is a \emph{synchronisation} that can happen only in cooperation between a system and
    its environment.}
An action in $\OA\dcup \{\tau\}$---so an output or internal action---is also called a \emph{non-blocking} action.

Now we formulate our \emph{progress assumption}\cite{TR13,GH14}:
\begin{equation*}
\parbox{0.87\textwidth}{\textit{%
Any process in a state that admits an non-blocking action will eventually perform an action.
}}
\end{equation*}

LTL formulas are deemed to hold for a process $P$ if they hold for all \emph{complete
paths} of $P$ in the labelled transition system of CCS$^!\!$. Here a \emph{path} of
$P$ is an alternating sequence of states and transitions, starting from the state $P$ and
either being infinite or ending in a state, such that each transition in the sequence goes from the
state before to the state after it, and a finite path is complete iff it does not end in a state that
enables a non-blocking action; completeness of infinite paths is discussed in the next
sections. For further details, see \cite[Section 9.1]{TR13} or~\cite{GH14}.

Assuming progress,
the scheduler $\Do \stackrel{{\it def}}{=} r.e.t.\Do$ 
satisfies $\mathbf{G}(r \Rightarrow \mathbf{F}(t))$ because on each complete
  path every $r$ is followed by a $t$. Hence $\Do$ is fair w.r.t.\ the simpler problem.

\section{A Just Semantics of Parallelism}\label{sec:justness}

In the previous section we considered a scheduler\vspace{2pt} that was significantly simpler than the one of
\Sect{fair scheduler}, and were able to specify it in CCS$^!\!$ by \plat{$\Do \stackrel{{\it def}}{=} r.e.t.\Do$}, with output actions $t$ and $e$. 
However, in order to ensure that our
specification was formally correct, we needed to introduce the concept of an output action, and made
a progress assumption on the semantics of the language.

In this section, we consider again a simplification of the scheduler of \Sect{fair scheduler}, and
once more succeed in specifying it in CCS$^!\!$. Again we need to make an assumption about the semantics
of CCS$^!$ in order to ensure that our specification is formally correct.

Both assumptions increase the range of correct CCS specifications and thereby make the promised
result on the absence of any CCS specification of a scheduler as described in \Sect{fair scheduler}
more challenging.

Consider a scheduler as described in \Sect{fair scheduler}, but without the last requirement about
the intermittent activity $e$.  A candidate CCS$^!$ specification is the process $\D_1 | \D_2$, defined by\vspace{-2ex}
\[ \D_i \stackrel{{\it def}}{=} r_i.t_i.\D_i\,,\quad i\in\{1,2\}.\]
Here, and throughout this paper, $t_i$ (like $e$) is an output action and $r_i$ is not.
For this scheduler to be fair, it has to satisfy
$\mathbf{G}(r_i \Rightarrow \mathbf{F}(t_i))$ for $i\mathbin=1,2$.\footnote{When assuming that these formulas hold,
  $\D_1|\D_2$ trivially satisfies the other properties required of it:
  the system will always return to a state where it remains ready to
  accept the next request $r_i$ until it arrives---hence for any numbers $n_1$ and
  $n_2\in\NN\cup\{\infty\}$ there is at least one run of the system in which exactly that many
  requests of type $r_1$ and $r_2$ are received---and in any partial run there are no more
  occurrences of $t_i$ than of $r_i$.}
By the reasoning of the previous section the process
$\D_i$ satisfies the temporal formula $\mathbf{G}(r_i \Rightarrow \mathbf{F}(t_i))$ for $i\mathbin=1,2$.
It is tempting to conclude that obviously their parallel composition $\D_{1} | \D_{2}$ satisfies both of
these requirements. Yet, the system run $r_1(r_2 t_2)\inf$\linebreak[1]---that after performing one action
from $\D_1$ performs infinitely many actions from $\D_2$ without interleaving any further actions from
$\D_1$---could be considered a counterexample.

Here we take the point of view that no amount of activity of $\D_2$ can prevent $\D_1$ from making
progress. The system $\D_1|\D_2$ simply does not have a run $r_1(r_2 t_2)\inf\!$.
The corresponding path from the state $\D_1|\D_2$ in the LTS of CCS$^!$ is no
more than an artifact of the use of interleaving semantics. In general, we make the following
\emph{justness assumption} \cite{GH14}: 
\[
\parbox{0.87\textwidth}{\textit{%
    If a combination of components in a parallel composition is in a state
    that admits a non-blocking action, then one (or more) of them  will eventually partake in an action.
    }}
\]
Thus justness guarantees progress of all components in a parallel composition, and of all
combinations of such components.\vspace{2pt} In the CCS$^!$ expression $((P|Q)\backslash a) | R$ 
with  
\plat{$P \stackrel{{\it def}}{=} a.P + c.P$}, 
\plat{$Q \stackrel{{\it def}}{=} \bar a.Q$} and 
\plat{$R \stackrel{{\it def}}{=} b.R + \bar c.R$}
 for
instance
there is a state where $P$ admits an action $c\mathbin\in\HC$ with $c\neq a$ and $R$ admits an action
$\bar c$. Thereby, the combination of these components admits an action $\tau$. Our justness
assumption now requires that either $P$ or $R$ will eventually
partake in an action.
This could be the $\tau$-action obtained from synchronising $c$ and $\bar c$, but it also could be
any other action involving $P$ or $R$. In each case the system will (at least for an instant) cease to be in a state where that
synchronisation between $P$ and $R$ is enabled.
Note that progress is a special case of justness, obtained by considering any process as the
combination of all its parallel components.

In \cite{GH14} we formalised the justness assumption as follows.\newline
Any transition \plat{$P|Q \ar{\alpha} R$} derives, through the
rules of \Tab{CCS}, from
\vspace{-1ex}
\begin{itemize}
\item a transition \plat{$P \ar{\alpha} P'$} and a state $Q$, where $R=P'|Q$\,,
\item two transitions \plat{$P \ar{\alpha_1} P'$ and $Q \ar{\alpha_2} Q'$}, where $R=P'|Q'$\,,
\item or from a state $P$ and a transition \plat{$Q \ar{\alpha} Q'$}, where $R=P|Q'$.
\vspace{-1ex}
\end{itemize}
This transition/state, transition/transition or state/transition pair is called a \emph{decomposition}
of \plat{$P|Q \ar{\alpha} R$}; it need not be unique, as we will show in \Ex{not unique} below.
Now a \emph{decomposition} of a path $\eta$ of $P|Q$ into paths $\eta_1$ and $\eta_2$ of
$P$ and $Q$, respectively, is obtained by \hypertarget{hr:decomp}{decomposing}\label{pg:decomp} each transition in the path, and
concatenating all left-projections into a path of $P$ and all right-projections into a
path of $Q$.
Here it could be that $\eta$ is infinite, yet either $\eta_1$ or $\eta_2$ (but not both) are finite.
Again, decomposition of paths need not be unique.

Similarly, any transition \plat{$P[f] \ar{\alpha} R$} stems from a transition \plat{$P \ar{\beta} P'$},
where $R=P'[f]$ and $\alpha=f(\beta)$.
This transition is called a decomposition of \plat{$P[f] \ar{\alpha} R$}. A \emph{decomposition}
of a path $\eta$ of $P[f]$ is obtained by decomposing each transition in the path, and
concatenating all transitions so obtained into a path of $P$.
A decomposition of a path of $P\backslash c$ is defined likewise.

\begin{definition}\label{df:just path}\rm
\emph{$Y\!$-justness}, for $Y\mathbin\subseteq\HC$,\footnote{By definition $Y$ does not contain non-blocking action.}
 is the largest family of predicates on the paths in the LTS of CCS$^!$ such that
\begin{itemize}
\vspace{-1ex}
\item a finite $Y\!$-just path ends in a state that admits actions from $Y$ only;
\item a $Y\!$-just path of a process $P|Q$ can be decomposed into an $X$-just path of $P$ and a $Z$-just
  path of $Q$ such that $Y\mathbin\supseteq X\mathord\cup Z$ and $X\mathord\cap \bar{Z}\mathbin=\emptyset$---here
  $\bar Z\mathbin{:=}\{\bar{c} \mid c\mathbin\in Z\}$;
\item a $Y\!$-just path of
  $P\backslash c$ can be decomposed into a $Y\mathord\cup\{c,\bar c\}$-just path of $P$;
\item a $Y\!$-just path of $P[f]$ can be decomposed into an
  $f^{-1}(Y)$-just path of $P$; 
\item and each suffix of a $Y\!$-just path is $Y\!$-just.
\vspace{-1ex}
\end{itemize}
A path $\eta$ is \emph{just} if it is $Y\!$-just for some $Y\subseteq\HC$.
It is \emph{$a$-enabled} for an action $a\in\HC$ if $a\in Y$ for all $Y$ such that $\eta$ is $Y\!$-just.
\end{definition}
Intuitively, a $Y\!$-just path models a run in which $Y$ is an upper~bound of the set of labels of
abstract transitions\footnote{The CCS process $a.0 | b.0$ has two transitions labelled $a$, namely
\plat{$a.0 | b.0 \ar{a} 0|b.0$} and \plat{$a.0 | 0 \ar{a} 0|0$}. The only difference between these two transitions
  is that one occurs before the action $b$ is performed by the parallel component and the other
  afterwards. In \cite{GH14} we formalise a notion of an \emph{abstract transition} that
  identifies these two concrete transitions.} that from some point onwards are continuously enabled but never taken.
Here an {abstract transition} with a label from $\HC$ is deemed to be continuously enabled but never
taken iff it is enabled in a parallel component that performs no further actions.
Such a run can actually occur if the environment from some point
onwards blocks the actions in $Y$.

The last clause in the second requirement prevents an $X$-just path of $P$ and a $Z$-just path of
$Q$ to compose into an $X\mathord\cup Z$-just path of $P|Q$ when $X$ contains an action $a$ and $Z$
the complementary action~$\bar a$. The reason is that no environment can block both actions for
their respective components, as nothing can prevent them from synchronising with each other.
The fifth requirement helps characterising processes of the form $b.0+(A|b.0)$ and $a.(A|b.0)$, with  \plat{$A \stackrel{{\it def}}{=} a.A$}.
Here, the first transition `gets rid of' the choice and of the leading action $a$, respectively, 
and reduces the justness of paths of such processes to their suffixes.

A complication in understanding \Def{just path} is that a single path could be seen as
modelling different system runs of which one could be considered just, respectively $a$-enabled, and the other not.
\begin{example}\label{ex:not unique}
Consider the process $B|B$ defined by \plat{$B \stackrel{{\it def}}{=} b.B$}.
The only transition of this process is \plat{$B|B \ar{b} B|B$}, so $B|B$ has exactly one infinite path $\eta$,
obtained by repeating this transition infinitely often. Assuming that $b$ is an output action, one
may wonder if $\eta$ should count as being just. In case all transitions in $\eta$ originate
from the left component, the $b$-transition of the right component is continuously enabled but
never taken. This does not correspond to a (just) run of the represented system.
However, in case $\eta$ alternates transitions from each component, it does model a (just) run.
The mere fact that a $b$-transition is enabled on every state of $\eta$
has no bearing on the matter. Now \Def{just path} considers $\eta$ just, on grounds of the fact that
it models some (just) run.

If in this example $b$ is a handshake action, the path $\eta$ models a (just) run in which a
$b$-labelled abstract transition is continuously enabled but never taken; but it also models a
(just) run in which no transition is continuously enabled but never taken.
According to \Def{just path}, $\eta$ counts as $\emptyset$-just, and thus is not deemed $b$-enabled.
Intuitively, a path is $b$-enabled iff on all runs modelled by that path a transition labelled $b$
is continuously enabled but never taken.
\qed
\end{example}

Now a just path, as defined above, is our default definition of a complete path, as contemplated at
the end of \Sect{progress}. Indeed, a finite path is just iff it does not
end in a state from which a non-blocking action is possible \cite{GH14}.

Thus, our semantics of a CCS$^!$ process $P$ consists of the state $P$ in the LTS of
CCS$^!$ together with the set of complete paths in that LTS starting from $P$ \cite{TR13,GH14}.
LTL formulas \emph{hold} for $P$ iff they are valid on all complete paths of $P$.

Here we employ a just semantics of CCS$^!$ by taking the just paths to be the complete ones.
This way $\Do$ is a correct specification of the scheduler required in \Sect{progress} and $\D_1|\D_2$
is a correct specification of the scheduler required above. 

\section{Formal Specification of the Fair Scheduler}\label{sec:fair spec}
 
\renewcommand{\d}[1]{c_{#1}}
\newcommand{\db}[1]{c_{{#1}}!}
\newcommand{\dr}[1]{c_{#1}?}
\newcommand{\tb}[1]{t_{#1}}
\newcommand{\tr}[1]{t_{#1}?}
\newcommand{\e}{e}
\newcommand{\BB}[1]{B_{#1}}

We now provide a formal specification of the scheduler described in \Sect{fair scheduler}.
Since the aim of this paper is to show that this cannot be done in CCS$^!$ (and thus certainly not
in CCS) we need a different formalism for this task. Here we follow the
traditional approach of TLA~\cite{TLA} and several other frameworks \cite{Fr86}, ``in which
first the legal computations are specified, and then a fairness notion
is used to exclude some computations which otherwise would be legal'' \cite{AFK88}.
Following \cite{GH14}, we use CCS$^!$ for the first step and LTL for the second.

Thus, in this section we specify a process as a pair of a CCS$^!$ expression $P$ and a set $\FS$ of LTL
formulas, called a \emph{fairness specification}. The semantics of $P$ consists of the state $P$ in
the LTS of CCS$^!$ together with the set of just paths in that LTS starting from $P$. 
The formulas of $\FS$ are evaluated on the paths of $P$ and any path that satisfies all
formulas of $\FS$ is called \emph{fair}. Now the semantics of the entire specification $(P,\FS)$ is the
state $P$ in the LTS of CCS$^!$ together with the set of \emph{complete paths} of $P$,
defined as those paths that are both just and fair.
In \cite{TR13,GH14} a consistency requirement is formulated that should hold between $P$ and $\FS$.

Now a fair scheduler as described in \Sect{fair scheduler} can be specified by the CCS$^!$ process
$
(I_{1}\,|\,G\,|\,I_{2})\backslash \d1\backslash\d2
$, 
where\\[.5mm]
\centerline{$
I_{i}\stackrel{\it def}{=} r_i.\bar{c_i}.I_{i}\,\quad (i\in\{1,2\})
\quad \mbox{and}\quad
G\stackrel{\it def}{=}c_1.\tb1.\e.G  \,+\, c_2.\tb2.\e.G \,,
$}\\[1mm]
augmented with the fairness specification
$
\bigwedge_{i=1,2}\mathbf{G}(r_i \Rightarrow \mathbf{F}(t_i))
$.

Here the first requirement of \Sect{fair scheduler} is satisfied by locating the two
channels receiving the requests $r_1$ and $r_2$ on different parallel components $I_{1}$ and $I_{2}$.
This way, after performing, say, $r_1$, the system---component $I_{1}$ to be precise---will always return
to a state where it remains ready to accept the next request $r_1$ until it arrives, independent of
occurrences of $r_2$.
The (non-output) actions $c_{i}$ are used to communicate the request from the processes $I_{i}$ to 
a central component $G$, which then performs the requested action $t_{i}$.

The second requirement of \Sect{fair scheduler} is enforced by the fairness specification,
and the last two requirements of \Sect{fair scheduler} are met by construction:
in any partial run there are no more occurrences of $t_i$ than of $r_i$,
and between each two occurrences of $t_i$ and $t_j$ for $i,j\in\{1,2\}$ the intermittent action $e$ is scheduled.

\section[Fair Schedulers Cannot be Rendered in CCS---Formalisation]
        {Fair Schedulers Cannot be Rendered in CCS$^!$---Formalisation}\label{sec:no fair scheduler}

In this section we formulate the main result of the paper, namely that no scheduler as described in
Sections~\ref{sec:fair scheduler} and~\ref{sec:fair spec} can be specified in CCS$^!$. Since we already showed that it \emph{can} be
specified in CCS$^!$ augmented with a fairness specification, here, and in the rest of the paper, we
confine ourselves to CCS$^!$ without fairness specifications. Thus, our notion of a
complete path is (re)set to that of a just path, as specified in \Def{just path}.

\begin{theorem}\rm\label{thm:no fair scheduler}
There does not exist a CCS$^!$ expression $F$ such that:
\vspace{-1ex}
\begin{enumerate}
\item any complete path of $F$ that has finitely many occurrences of $r_i$ is $r_i$-enabled;\label{enabled}
\item on each complete (= just) path of $F$, each $r_i$ is followed by a $t_i$;\label{fair}
\item on each finite path of $F$ there are no more occurrences of $t_i$ than of $r_i$;\label{patient} and
\item between each two occurrences of $t_i$ and $t_j$ ($i,j\in\{1,2\}$) an action $\e$ occurs.\label{coordinated}
\end{enumerate}
\end{theorem}
Requirements \ref{enabled}--\ref{coordinated} exactly formalise the
four requirements described in \Sect{fair scheduler}. We proceed to show that none of them can be skipped.

The CCS$^!$ process $\D_1|\D_2$ of \Sect{justness}
satisfies Requirements \ref{enabled}, \ref{fair} and \ref{patient}.
It does not satisfy Requirement \ref{coordinated}, due to the partial run $r_1r_2t_1t_2$.
\vspace{2pt}

The CCS$^!$ process $E_1|G_1|E_2$ with \plat{$E_i \stackrel{\it def}{=} r_i.E_i$} for $i\mathbin=1,2$ and
\plat{$G_1 \stackrel{\it def}{=} t_1.e.t_2.e.G_1$}
satisfies Requirements \ref{enabled}, \ref{fair} and \ref{coordinated}.
It does not satisfy Requirement \ref{patient}, due to the partial run consisting of the single action $t_1$.

The CCS$^!$ process $E_1|E_2$ satisfies
Requirements \ref{enabled}, \ref{patient} and \ref{coordinated}, but not \ref{fair}.
\vspace{2pt}

Finally, the process $G_2$ given by
\plat{$G_2\stackrel{\it def}{=} r_1.\tb1.\e.G_2  \,+\, r_2.\tb2.\e.G_2$}
satisfies Requirements \ref{fair}, \ref{patient} and \ref{coordinated}. However, it does not
satisfy Requirement \ref{enabled}, because it allows the $\emptyset$-just path $(r_2 t_2 e)\inf$
with no occurrences of $r_1$. This path models a run in which the system never reaches a state where
it \emph{remains} ready to accept the next request $r_1$.

The proof of \Thm{no fair scheduler} will be the subject of
\Part{three}.

\section[A Characterisation of Fair Schedulers without enabling]{A Characterisation of Fair Schedulers without $a$-enabling}\label{sec:characterisation}

Below we will show that without loss of generality we may assume
any fair scheduler to have a specific form. If it has that form,
Requirement~\ref{enabled} is redundant. Hence Requirement~\ref{enabled} can be replaced by
requiring that the scheduler is of that form.
This variant of our result appeared as a conjecture in \cite{GH14}.

For any CCS$^!$ expression $F$, let $\widehat F:= (I_1\,|\,F[f]\,|\,I_2)\backslash \d1\backslash\d2$
with \plat{$I_{i}\stackrel{\it def}{=} r_i.\bar{c_i}.I_{i}$} for $i\in\{1,2\}$,
where $f$ is an injective relabelling with $f(r_i)=c_i$ for $i=1,2$, and
${r_1},{r_2},\bar{r_1},\bar{r_2}\notin f(\HC)$.
By the definition of relabelling (cf. Section~\ref{sec:progress}), $f(t_{i})=t_{i}$ and $f(e)=e$.

\begin{theorem}\rm\label{thm:no fair scheduler2}
A process $F$ meets Requirements \ref{enabled}--\ref{coordinated} of \Thm{no fair scheduler} iff
$\widehat F$ meets these requirements, which is the case iff
$\widehat F$ meets Requirements \ref{fair}--\ref{coordinated}.
\end{theorem}
\begin{proof}
Suppose $F$ satisfies Requirements \ref{enabled}--\ref{coordinated}.
\begin{enumerate}
\item To show that $\widehat F$ satisfies Requirement \ref{enabled} (with $i\mathbin{:=}1$; the other case
  follows by symmetry), it suffices to show that each occurrence of $r_1$ in a just path of
  \plat{$\widehat F$}, which corresponds to an occurrence of $r_1$ in the subprocess $I_1$, is followed by
  an occurrence of $\bar{c_1}$ in $I_1$.

  So assume, towards a contradiction, that on a just path $\eta$ of \plat{$\widehat F$} an
  occurrence of $r_1$ is not followed by an occurrence of $\bar{c_1}$ in the subprocess $I_1$.
  By \Def{just path} $\eta$ must be $Y\!$-just for some $Y\subseteq\HC$.
  So $\eta$ can be decomposed into an $X$-just path $\eta_1$ of $I_1$, a $Z$-just path $\eta_0$ of
  $F[f]$ and a $W$-just path $\eta_2$ of $I_2$ for certain $X,Z,W\subseteq\HC$. By assumption, $\bar{c_1}\in X$.
  Moreover, $\eta_0$ can be decomposed into an $f^{-1}(Z)$-just path $\eta_F$ of $F$.
  Since in \plat{$\widehat F$} the $c_1$ of $F[f]$ requires synchronisation with the
  $\bar{c_1}$ of $I_1$, and $\eta_1$ has only finitely many occurrences of $\bar{c_1}$,
  it follows that $\eta_0$ has only finitely many occurrences of $c_1$, and thus that
  $\eta_F$ has only finitely many occurrences of $r_1$.
  Since $F$ satisfies Requirement~\ref{enabled}, saying that the system will always
  return to a state where it remains ready to accept the next request $r_1$ until it arrives,
  $r_1\in f^{-1}(Z)$. Hence $c_1\in Z$.
  By \Def{just path}, this contradicts the justness of $\eta$.
\item Above we have shown that each occurrence of $r_1$ in a just path of
  \plat{$\widehat F$}, which corresponds to an occurrence of $r_1$ in the subprocess $I_1$, is followed by
  an occurrence of $\bar{c_1}$ in $I_1$. This occurrence of $\bar{c_1}$ in $I_1$ must be a
  synchronisation with an occurrence of $c_1$ in $F[f]$, and by Requirement~\ref{fair} for $F$ each
  occurrence of $c_1$ in $F[f]$ is followed by an occurrence of $t_1$ in $F[f]$, and hence in $\widehat{F}$.
\item By Requirement~\ref{patient} for $F$, on each finite path from $F[f]$ there are no more
  occurrences of $t_1$ than of $c_1$. Moreover, on each finite path from $I_1$ there are no more
  occurrences of $\bar{c_1}$ than of $r_1$. Since in $\widehat F$ each occurrence of $c_1$ in $F[f]$ needs to
  synchronise with an occurrence of  $\bar{c_1}$ in $I_1$, it follows that on each finite path from
  $\widehat F$  there are no more occurrences of $\bar{c_1}$ than of $r_1$.
\item Requirement~\ref{coordinated} holds for $\widehat F$ because it holds for $F$.
\end{enumerate}
Now assume that $\widehat F$ satisfies Requirements \ref{fair}--\ref{coordinated}.
\begin{enumerate}
\item Suppose that $F$ would fail Requirement~\ref{enabled}, say for $i=1$. Then it has a $Z$-just
  path with $r_1\notin Z$. Therefore $F[f]$ has an $f(Z)$-just path with $c_1\notin f(Z)$.
  This path can be synchronised with a $\bar{c_1}$-just path of $I_1$ into a just path of $\widehat{F}$
  in which an occurrence of $r_1$ follows the last occurrence of $t_1$, thereby violating
  Requirement~\ref{fair} for $\widehat{F}$.
\item Suppose that $F$ would fail Requirement~\ref{fair}, say for $i=1$. Then it has a just
  path with an occurrence of $r_1$ past the last occurrence of $t_1$. Therefore $F[f]$ has a $Z$-just
  path with $\bar{r_1}\notin Z$ and an occurrence of $c_1$ past the last occurrence of $t_1$.
  This path can be synchronised with a $r_1$-just path of $I_1$ into a just path of $\widehat{F}$
  in which an occurrence of $r_1$ follows the last occurrence of $t_1$, thereby violating
  Requirement~\ref{fair} for $\widehat{F}$.
\item Suppose $F$ had a finite path with more occurrences of $t_1$ than of $r_1$, then through
  synchronisation a finite path of $\widehat{F}$ could be constructed with more occurrences of $t_1$ than of $r_1$.
\item Requirement~\ref{coordinated} holds for $F$ because it holds for $\widehat F$.
\qed
\end{enumerate}
\end{proof}
Recall that $G_2$, given by $G_2\stackrel{\it def}{=} r_1.\tb1.\e.G_2  \,+\, r_2.\tb2.\e.G_2$,
satisfies Requirements \mbox{\ref{fair}--\ref{coordinated}}. 
Converting this $G_2$ to the process $\widehat {G_2}$ of the form  $(I_1\,|\,G\,|\,I_2)\backslash
\d1\backslash\d2$, as defined above, results in the specification of \Sect{fair spec}
without the additional fairness specification, and hence in the loss of Requirement \ref{fair}.

\section{\hspace{-0.4pt}(\hspace{-0.1pt}In\hspace{-0.1pt})\hspace{-0.4pt}Correct\hspace{-0.4pt} Correctness\hspace{-0.4pt} Proofs\hspace{-0.4pt} of\hspace{-0.4pt} Peterson's\hspace{-0.4pt} and\hspace{-0.4pt} Dekker's\hspace{-0.4pt} Protocols}\label{sec:discussion}

It is widely accepted that Peterson's mutual exclusion protocol~\cite{Peterson81} implements a fair
scheduler, and that implementing Peterson's algorithm in a CCS-like language should be easy.
In fact Peterson's algorithm has been specified in CCS-like languages several times, e.g.\
\cite{Walker89,Bouali91,Valmari96,AcetoEtAl07}.   All these papers present essentially the same rendering of Peterson's algorithm
in CCS or some other progress algebra, differing only in insignificant implementation details.
This seems to contradict our main result (Theorem~\ref{thm:no fair scheduler}).
\newcommand{\procA}{{\rm A}\xspace}
\newcommand{\procB}{{\rm B}\xspace}

Peterson's Mutual Exclusion Protocol deals with two concurrent processes \procA and \procB that want to
alternate critical and noncritical sections. Each of these processes
will stay only a finite amount of time in its critical section,
although it is allowed to stay forever in its noncritical section.
The purpose of the algorithm is to ensure that they are never
simultaneously in the critical section, and to guarantee that both
processes keep making progress. 
Pseudocode is depicted in \Fig{peterson}.

\begin{figure}
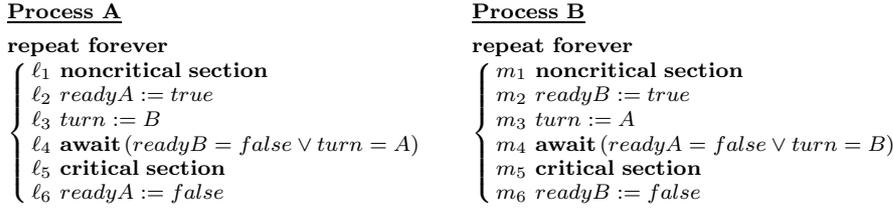

\[
\begin{array}{@{}l@{}}
\underline{\bf Process~A}\\[1ex]
{\bf repeat~forever}\\
\left\{\begin{array}{ll}
\ell_1 & {\bf noncritical~section}\\
\ell_2 & readyA := true	\\
\ell_3 & turn := B \\
\ell_4 & {\bf await}\,(readyB = false \vee turn = A) \\
\ell_5 & {\bf critical~section}\\
\ell_6 & readyA := false \\
\end{array}\right.
\end{array}
~~~~~~
\begin{array}{@{}l@{}}
\underline{\bf Process~B}\\[1ex]
{\bf repeat~forever}\\
\left\{\begin{array}{ll}
m_1 & {\bf noncritical~section}\\
m_2 & readyB := true	\\
m_3 & turn := A \\
m_4 & {\bf await}\,(readyA = false \vee turn = B) \\
m_5 & {\bf critical~section}\\
m_6 & readyB := false \\
\end{array}\right.
\end{array}
\]
\caption{Peterson's algorithm (pseudocode)}
\label{fig:peterson}
\end{figure}

The processes use three variables. The Boolean variable $readyA$ can be
written by process \procA and read by process \procB, whereas $readyB$ can be
written by \procB and read by \procA. By setting $readyA$ to $true$, process
\procA signals to process \procB that it wants to enter the critical
section. The variable $turn$ is a shared variable: it can be written
and read by both processes. Its use is the brilliant part of the algorithm.
Initially $readyA$ and $readyB$ are both $false$ and $turn = A$.

Peterson's algorithm implements a mutual exclusion protocol and hence 
should satisfy the safety property that at any time only one process accesses the critical system, i.e.\
\[{\bf G} (\neg((\ell_{4}\vee\ell_{5})\wedge (m_{4}\vee m_{5})))\,.\]
Here, $\ell_{i}$ and $m_{j}$ refer to line numbers of the pseudocode (\Fig{peterson}). 
As convention we assume that line numbers refer to a state in the execution of the code where the 
command of the line has already been executed.
Most papers, including \cite{Bouali91,AcetoEtAl07},  
concentrate on the issue of mutual
exclusion only, and that is done correctly in the CCS rendering.
When safety properties are considered only, no fairness or progress assumption is needed: 
in the worst case some (or all) processes do not progress and hence never enter the critical section---%
the safety property still holds. 

As usual, a safety property should therefore be accompanied with a liveness property.
In case of Peterson's protocol such a property is that any process
that wants to enter the critical section will at
some point reach the critical section. We consider two possibilities to characterise this property:
\begin{eqnarray*}
\bf{G}(\ell_{1}\Rightarrow{\bf F}(\ell_{4}))\,,\quad\mbox{and}\qquad
\bf{G}(\ell_{2}\Rightarrow{\bf F}(\ell_{4}))\,.
\footnotemark
\end{eqnarray*}
\footnotetext{We only give the liveness property for process \procA; the one for process \procB is similar.}%
Both properties have the form ${\bf G} (r_{1}\Rightarrow {\bf F}(t_1))$---the property discussed in this paper---%
$\ell_{4}$ indicates that the process enters the critical section. Both $\ell_{1}$ and $\ell_{2}$ could 
play the r\^ole of the grant request $r_{1}$.
Although it seems surprising, we will show that there is a fundamental difference between these two formulas.

To show how Peterson's algorithm yields an instance of our fair scheduler, we consider the
action $t_1$ to be taken when an execution passes through state $\ell_4$, thereby interpreting $t_1$ as
granting access to the critical section. The action $e$ is taken when the execution passes through state
$\ell_5$, marking the \emph{exit} of the critical section. Peterson's code, in combination
with the mutual exclusion property, ensures that Requirement~\ref{coordinated} of our fair scheduler
is satisfied. We consider $r_1$ to be taken when an execution passes through state $\ell_1$, so that the
liveness property ${\bf{G}}(\ell_{1}\Rightarrow{\bf F}(\ell_{4}))$ ensures Requirement~\ref{fair}.
Requirement \ref{enabled} is satisfied, because as soon as the environment of the protocol leaves
the noncritical section, thereby getting ready to enter the critical one, the protocol is considered
to take the action $r_1$. Finally, Requirement~\ref{patient} is obviously ensured by Peterson's code.

In combination with this insight, our main result (\Thm{no fair scheduler}) entails that
the rendering of Peterson's algorithm in CCS found in the literature cannot be correct, as long as the semantics of CCS
is fortified with at most justness.
To prove liveness of Peterson's protocol,
 at least weak fairness  is required.\footnote{Whether weak fairness suffices depends on the interpretation of enabledness (cf.\ \Sect{fairness})}
In the literature we found only two papers that investigate
liveness properties of this protocol:
\cite{Walker89} and \cite{Valmari96}. Neither of these papers employs fairness or justness properties.

Walker~\cite{Walker89} tries to prove the correctness of Peterson's algorithm by automatic methods.
He succeeds for the safety property, but could not establish the liveness property
in full generality; however Walker succeeded in
proving it when restricting attention to runs in which infinitely many
visible actions occur. This appears to be Walker's method of imposing a progress assumption.
Although this is strictly speaking not in
contradiction with our results, our proofs trivially extend to the
case of considering only runs in which infinitely many visible actions
occur. Hence Walker's result seems to be in contradiction to 
\Thm{no fair scheduler}.
A detailed analysis reveals that Walker uses line $\ell_{2}$ as request action to indicate 
interest to enter the critical section.
So he shows that $\bf{G}(\ell_{2}\Rightarrow{\bf F}(\ell_{4}))$.
That means that the \emph{shared variable} $readyA$ must 
be set---only then $\ell_{2}$ evaluates to true.
His request action is set right after
setting this variable. 
However, following our proof, the reason that
the CCS rendering of Peterson does not work, is that it is possible that
process \procA never gets a change to set the shared variable $readyA$ to true, because the
other process is too busy reading it all the time (even when it
enters the critical section between any two reads). 
So, it is a possible scenario that process \procA will never execute line $\ell_{2}$, although 
it wants to enter the critical section.

In Peterson's
original thinking, process \procB could not prevent process \procA
from writing by reading a shared variable; but in the CCS model this is
quite possible: the read action can only be represented as a transition that is in conflict with the write action;
only after this transition is taken does the process return to a state where the write is enabled.
So when Walker~\cite{Walker89} establishes $\bf{G}(\ell_{2}\Rightarrow{\bf F}(\ell_{4}))$
he merely shows that when the first hurdle is taken successfully
the process will surely enter the critical section. What he cannot establish
is that a process that is ready to enter the
critical section will succeed in setting $readyA$. 
The correct modelling of Peterson's liveness property
thus places action $r_1$ \emph{before} setting the variable $readyA$ to true, i.e.\
\[
\bf{G}(\ell_{1}\Rightarrow{\bf F}(\ell_{4}))\,.
\]
In terms of our description of a fair scheduler, the action $r_1$ of
Walker (at position $\ell_2$) does not meet Requirement~\ref{enabled}.

The analysis of the work of Walker shows that there is a fine line between 
correct and incorrect modelling. In fact it looks reasonable to prove 
$\bf{G}(\ell_{2}\Rightarrow{\bf F}(\ell_{4}))$ instead of 
$\bf{G}(\ell_{1}\Rightarrow{\bf F}(\ell_{4}))$. There is
no formal way to avoid such mistakes; only careful (informal) reasoning.

Roughly the same modelling, but in which the
request $r_1$ is \emph{identified}
with setting the shared variable $readyA$ to true,
occurs in Valmari \mand Set\"{a}l\"{a}~\cite{Valmari96}.
The consequences are the same.

Dekker's mutual exclusion protocol~\cite{EWD35,EWD123} is another well-known algorithm that implements a fair
scheduler. We found two papers in the literature that analyse liveness of this protocol.

Esparza \mand Burns~\cite{EsparzaBruns96}
follow in the footsteps of Walker and prove the correctness of
Dekker's mutual exclusion algorithm in the Box Calculus without
  postulating a fairness assumption. According to our results, this
is impossible as well. Indeed, as in \cite{Valmari96}, Esparza \mand Burns model the request $r_1$
to be the action of setting a shared variable, which again violates the property
that a process that wants to enter the critical section can always
succeed at least in making a request to that effect.

Corradini, Di Berardini \mand Vogler~\cite{CorradiniEtAl09}
specify Dekker's algorithm in the CCS-like process algebra PAFAS. 
They also prove the correctness 
of the algorithm. This paper models the relevant liveness properties correctly, as far
as we can see, but explicitly makes different assumptions on the driving
force that keeps the system running. First they consider a notion of
`fairness of actions' that appears to be similar to our
justness assumption,\footnote{It differs in a crucial way,
    however, namely by treating each action as output. As a
    consequence, under fairness of actions the process $F_1|F_2$ of
    \Sect{justness} is guaranteed to perform each of the actions $r_1$
    and $r_2$ infinitely often.  To model a protocol where the action
    $r_i$ is not forced to occur, a $\tau$-loop is inserted 
    at each location where $r_{i}$ is enabled.}
    and they show that their model of Dekker's protocol fails
to have the required liveness property. This result is entirely
consistent with ours. In fact we generalise their negative result
about the correctness of a particular rendering in PAFAS of a
particular protocol for mutual exclusion to a general statement
quantifying of all renderings of all such protocols.

Next they consider a stronger notion of fairness called `fairness
of components', stemming from \cite{CS87}, and, under this assumption, establish the correctness of
the algorithm.%
\footnote{Fairness of components is a form of weak fairness,
requiring that if a \emph{component} from some point onwards is enabled in each state, an action
from that component will eventually be scheduled.
Here a component is enabled if an action from that component is enabled, possibly in synchronisation
with an action from outside that component.\vspace{2pt} Under this notion of fairness, the system $E$
from \Sect{fairness}, defined by \plat{$E \stackrel{{\it def}}{=} a.E + b.0$}, is not ensured to do
a $b$ eventually. However, the composition $(E|\bar b.c.0)\backslash b$ is ensured to do a $c$
eventually, because the component $\bar b.c.0$ is enabled in every state.}
The present paper augments this result by saying that a fairness notion as strong as `fairness
of components' is actually needed. 

In \Part{one} we pointed out that our result holds for CCS+justness, CCS+progress and CCS without any progress assumption. 
However a fair scheduler can be implemented when a fairness assumption is assumed; fairness of
components appears to be sufficient.

\part{Proofs}{three}
\section{Fair Schedulers Cannot be Rendered in Petri Nets---Formalisation}\label{sec:nets}

\newcommand{\Act}{Act}
This section introduces Petri nets and rephrases \Thm{no fair scheduler} in terms of Petri nets.
We inherit the sets $Act$ of actions and $\HC\subseteq\Act$ of handshaking communications from \Sect{CCS},
and the set $\OA$ of output actions from \Sect{progress}.

A {\em multiset} over a set $X$ is a function $A\!:X \rightarrow \NN$,
i.e.\ $A\in \NN^{X}\!$.
The function $\emptyset\!:X\rightarrow\NN$, given by
  $\emptyset(x):=0$ for all $x \in X$, is the \emph{empty} multiset over $X$.\\
$x \in X$ is an \emph{element of} $A$, notation $x \in A$, iff $A(x) > 0$.\\
For multisets $A$ and $B$ over $X$ we write $A \leq B$ iff
 \mbox{$A(x) \leq B(x)$} for all $x \in X$;
\\ $A\cap B$ denotes the multiset over $X$ with $(A\cap B)(x):=\text{min}(A(x), B(x))$,
\\ $A + B$ denotes the multiset over $X$ with $(A + B)(x):=A(x)+B(x)$, and
\\ $A - B$ is only defined if $B\leq A$ and then denotes the multiset over $X$ with 
$(A - B)(x):=A(x)-B(x)$.
\\ A multiset $A$ with $A(x)\leq 1$ for all $x$ is identified with the (plain) set $\{x\mid A(x)\mathbin=1\}$.%
\begin{definition}
  A (\emph{labelled}) \emph{Petri net} (\emph{over $\Act$}) is a tuple
  $N = (S, T, F, M_0, \ell)$ with
  \begin{list}{{\bf --}}{\leftmargin 18pt
                        \labelwidth\leftmargini\advance\labelwidth-\labelsep
                        \topsep 0pt \itemsep 0pt \parsep 0pt}
    \item $S$ and $T$ disjoint sets (of \emph{places} and \emph{transitions}),
    \item $F: ((S \times T) \cup (T \times S)) \rightarrow \NN$
      (the \emph{flow relation} including \emph{arc weights}),
    \item $M_0 : S \rightarrow \NN$ (the \emph{initial marking}), and
    \item \plat{$\ell: T \rightarrow \Act$} (the \emph{labelling function}).
  \end{list}
\end{definition}

\noindent
When a Petri net represents a concurrent
system, a global state of this system is given as a \emph{marking},
a multiset $M$ of places. 
The initial state is $M_0$.

The behaviour of a Petri net is defined by the possible moves between
markings $M$ and $M'$, which take place when a transition $u$ \emph{fires}.  In that case,
$u$ consumes $F(s,u)$ tokens from each 
place $s$.  Naturally, this can happen only if $M$ makes all these
tokens available in the first place.  Moreover, $u$ produces $F(u,s)$ tokens
in each $s$.  \Def{firing} formalises this notion of behaviour.

\begin{definition}\label{df:firing}
Let $N = (S, T, F, M_0, \ell)$ be a Petri net and $u\mathbin\in T$.
The multisets $\precond{u},~\postcond{u}: S \rightarrow
\NN$ are given by $\precond{u}(s)=F(s,u)$ and
$\postcond{u}(s)=F(u,s)$ for all $s \mathbin\in S$;%
\footnote{Here, we slightly deviate from
  standard notation \cite{Re85}, where $\precond{u}$ and $\postcond{u}$ are usually plain sets,
obtained from our multisets by abstracting from the multiplicities of their elements. 
We prefer to retain this information, so as to shorten various formulas.}
the elements of $\precond{u}$ and $\postcond{u}$ are
called \emph{pre-} and \emph{postplaces} of $u$, respectively.
Transition $u\mathbin\in T$ is \emph{enabled} from the marking $M\mathbin\in\NN^S$---notation
$M[u\rangle$---if $\precond{u} \leq M$.
In that case firing $u$ yields the marking 
$M':=M-\precond{u}+\postcond{u}$\linebreak[3]---notation $M[u\rangle M'$.
\end{definition}
A \emph{path} $\pi$ of a Petri net $N$ is an alternating sequence $M_0 u_1 M_1 u_2 M_2 u_3 \dots$ of markings and
transitions, starting from the initial marking $M_0$ and either being infinite or ending in a
marking $M_n$, such that $M_k [u_k\rangle M_{k+1}$ for all $k \,(\mathord< n)$.
An action $\alpha\in\Act$ \emph{occurs} on a path $\pi$ if there is a transition $u_i$ with $\ell(u_i)=\alpha$.
A marking is \emph{reachable} if it occurs in such a path.
The Petri net $N$ is
\emph{safe} if all reachable markings $M$ are plain sets, meaning that $M(s)\leq 1$ for all places $s$.
It is a \emph{structural conflict net} \cite{GGS11} if $\precond{u}+\precond{v}\leq M \Rightarrow
\precond{u}\cap\precond{v}=\emptyset$ for all reachable markings $M$ and all transitions $u$ and $v$.
Note that any safe Petri net is a structural conflict net.
In this paper we restrict attention to structural conflict nets with the additional assumptions that $\precond{u}\mathbin=\emptyset$ for no transition $u$, 
and that all reachable markings are finite.
In the remainder we refer
to these structures as \emph{nets}. For the purpose of establishing \Thm{no fair scheduler} we could
just as well have further restricted attention to
safe Petri nets whose reachable markings are finite.

On a path $\pi = M_0 u_1 M_1 u_2 M_2 u_3 \dots$ a transition $v$ is \emph{continuously enabled from position $k$ onwards} if
$M_k [v\rangle$ and $\precond{v}\cap\precond{u_i}=\emptyset$ for all $i \mathbin > k$.
This implies that $\precond{v}\leq M_i$ for all $i \geq k$.
If such a transition $v$ exists we say that $\pi$ is $\ell(v)$-enabled.
A path is \emph{just} or \emph{complete} if it is $o$-enabled for no non-blocking action $o\in\OA\cup\{\tau\}$.

Now we have all the necessary definitions to state that our fair scheduler cannot be realised as a net.

\begin{theorem}\rm\label{thm:no fair PN scheduler}
There does not exist a net $N$ such that:
\vspace{-1ex}
\begin{enumerate}
\item any complete path of $N$ that has finitely many occurrences of $r_i$ is $r_i$-enabled;\label{enabled PN}
\item on each complete (= just) path of $N$, each $r_i$ is followed by a $t_i$;\label{fair PN}
\item on each finite path of $N$ there are no more occurrences of $t_i$ than of $r_i$;\label{patient PN} and
\item between each two occurrences of $t_i$ and $t_j$ ($i,j\in\{1,2\}$) an action $\e$ occurs.\label{coordinated PN}
\end{enumerate}
\end{theorem}
In the proof of this theorem we do not use the restriction that $N$ is a structural conflict net.
However, for general Petri nets our definition of a transition that from some points onwards is
continuously enabled is not convincing. A better definition would replace the requirement 
$\precond{v}\cap\precond{u_i}=\emptyset$ for all $i>k$ by
$\precond{v}+\precond{u_{i+1}} \leq M_i$ for all $i\geq k$.
On structural conflict nets the two definitions are equivalent.
On general
Petri nets with finite reachable markings and $\forall u. \precond{u}\neq\emptyset$ \Thm{no fair PN scheduler}
still holds when employing our earlier definition of being continuously enabled, but that definition
arguably leads to Requirement~\ref{enabled PN} being an overly restrictive formalisation of the first requirement
of \Sect{fair scheduler}.

In \cite{Vogler02} and in \cite{KW97} \emph{mutex problems} are presented that cannot be solved
in terms of Petri nets. These results are almost equivalent to \Thm{no fair PN scheduler}, but,
as discussed in the \SSect{contribution}, lack the generality needed to infer \Thm{no fair scheduler} from \Thm{no fair PN scheduler}.

\section{Fair Schedulers Cannot be Rendered in Petri Nets---Proof}
\label{sec:proof PN}

In this section we suppose that there exists a net $N$
meeting the requirements of \Thm{no fair PN scheduler}.
We establish various results about this hypothetical net $N$, ultimately leading to a
contradiction. This will constitute the proof of \Thm{no fair PN scheduler}.

\subsection{Embellishing Paths into Complete Paths}

\newcommand{\fs}{\mbox{\sc fs}}
\newcommand{\ph}{\mbox{\sc ph}}
A \emph{firing sequence} of $N$ is a sequence $\sigma = u_1u_2u_3\dots$ of transitions such that
there exists a path $\pi = M_0 u_1 M_1 u_2 M_2 u_3 \dots$ of $N$. Note that $\pi$ is uniquely
determined by~$\sigma$ (and $M_{0}$); we call it $\ph(\sigma)$. Likewise, $\sigma$ is determined by $\pi$, and
we call it $\fs(\pi)$.

A firing sequence $\sigma'=v_1 v_2 v_3 \dots$ \emph{embellishes} a firing sequence $\sigma = u_1 u_2 u_3 \dots$
iff $\sigma'$ can be obtained out of $\sigma$ through insertion of non-blocking transitions;
that is, if there exists a monotone increasing function $f:\NN\rightarrow\NN$---thus satisfying $i\mathbin<j
\Rightarrow f(i)\mathbin<f(j)$---with $v_{f(i)}\mathbin=u_i$ for all $i\mathbin>0$ and $\ell(v_j)\mathbin\in\OA\cup\{\tau\}$ for any
index $j$ not of the form $f(i)$. A path $\pi'$ \emph{embellishes} a path $\pi$ iff $\fs(\pi')$
embellishes $\fs(\pi)$.

Given a firing sequence $\sigma\mathop=u_1u_2\dots$ of length $\mathord\geq k$ and a transition~$w$, let
$\sigma\oplus_kw$ denote the sequence $u_1 u_2 \dots u_{k}w u_{k+1} u_{k+2} \dots$ obtained by
inserting $w$ in $\sigma$ at position $k$.

\begin{lemma}\label{lem:one step embellishment}
Let $\sigma$ be a firing sequence of length $\mathord\geq k$ and $w$ a transition that on $\ph(\sigma)$ is continuously
enabled from position $k$ onwards. Then $\sigma\oplus_kw$ is a firing sequence.
\end{lemma}
\begin{proof}
Let $\ph(\sigma) = M_0 u_1 M_1 u_2 M_2 u_3 \dots$. Define
$M'_{i}:=M_i-\precond{w}+\postcond{w}$ for $i\geq k$. Then
$M_0 u_1 M_1 u_2 \dots u_{k} M_k w M'_k u_{k+1} M'_{k+1} u_{k+2} \dots$ is again a path of $N$,
using that $\precond{w}\leq M_i$ and $\precond{u_{i+1}} \leq M_i-\precond{w}$ for all $i\geq k$.
\qed
\end{proof}
If $\pi$ is a path and $w$ a transition that on $\pi$ is continuously enabled from position $k$
onwards, then $\pi\oplus_{k}w$ abbreviates $\ph(\fs(\pi)\oplus_{k}w)$.

A path $\pi = M_0 u_1 M_1 u_2 M_2 u_3 \dots$ of $N$ is \emph{$(k,n)$-incomplete} if $k$ is the
smallest number such that there is a transition $w$ with $\ell(w)\in\OA\cup\{\tau\}$---called a \emph{witness} of the
$(k,n)$-incompleteness of $\pi$---that is continuously enabled from position $k$ onwards, and $n$ is
the number of places $s$ of $N$ such that
$s \mathbin\in\precond{w}$ for a witness $w$ of the $(k,n)$-incompleteness of $\pi$.
Since the reachable marking $M_k$ is always finite, so are the numbers $n$.
Note that a path is $(k,n)$-incomplete for some finite $k$ and $n$ iff it is not complete;
henceforth we call a complete path \emph{$(\infty,0)$-incomplete}.
If a path $\pi$ is $(k,n)$-incomplete, and a path $\rho$ is $(h,m)$-incomplete, then we call
$\rho$ \emph{less incomplete} than $\pi$ if $\rho$ has the same prefix up to position $k$
as $\pi$ and either $h>k$ or $h=k \wedge m<n$.

\begin{lemma}\label{lem:less incomplete}
Let $i\mathbin\geq 0$, $\pi$ be a $(k,n)$-incomplete path of the net $N$ with at least $k\mathord+i$
transitions, and $w$ a witness of the $(k,n)$-incompleteness of $\pi$. Then $\pi\oplus_{k+i}w$ is less incomplete than $\pi$.
\end{lemma}
\begin{proof}
Suppose $\pi\oplus_{k+i}w$ is $(h,m)$-incomplete with $h\leq k$, and let $v$ be a witness of the
$(h,m)$-incompleteness of $\pi\oplus_{k+i}w$.
Then on $\pi\oplus_{k+i}w$ the transition $v$ is continuously enabled from position $h$ onwards.
Let $M_h$ be the marking occurring at position~$h$ in $\pi$, or equivalently in $\pi\oplus_{k+i}w$.
Then $\precond{v}\leq M_h$ and \mbox{$\precond{v}\cap\precond{u}\mathbin=\emptyset$} for all transitions
$u$ occurring in $\pi\oplus_{k+i}w$ past position $h$. This includes all transitions $u$ occurring
in $\pi$ past position $h$, so $v$ is continuously enabled from position $h$ onwards also on $\pi$.
It follows that $h=k$ and any witness of the $(k,m)$-incompleteness of $\pi\oplus_{k+i}w$ is also a
witness of the witness of the $(k,n)$-incompleteness of $\pi$.
Moreover, $\precond{v}\cap\precond{w}=\emptyset$, and since
$\precond{w}\neq\emptyset$ this implies $m<n$.
\qed
\end{proof}

\begin{lemma}\label{lem:embellishment by complete path}
Any infinite path in $N$ is embellished by a complete path.
\end{lemma}

\begin{proof}
Let $\pi$ be the given path.
We build a sequence $\pi_i$ of paths in $N$ that all embellish $\pi$, such that, for all $i$,
$\pi_{i+1}$ is less incomplete than $\pi_i$ and the first $2i$ transitions of $\pi_i$ and $\pi_{i+\!1}$ are the same.

We start by taking $\pi_0$ to be $\pi$.
If at any point we hit a path $\pi_i$ that is complete, our work is done.
Otherwise, given the $(k,n)$-incomplete path $\pi_i$, for some $k$ and $n$,
pick a witness $w$ of the $(k,n)$-incompleteness of $\pi_i$ and take $\pi_{i+1} \mathbin{:=} \pi_i\oplus_{k+2i}w$.
This path exists by \Lem{one step embellishment}, since $w$ is continuously enabled from position
$k$ onwards, and hence also from position $k+2i$ onwards.
By construction $\pi_{i+1}$ embellishes $\pi_i$ and hence $\pi$.
By \Lem{less incomplete} $\pi_{i+1}$ is less incomplete than $\pi_i$.
Moreover, the first $2i$ transitions of $\pi_i$ and $\pi_{i+\!1}$ are the same.

If at no point we hit a path $\pi_i$ that is complete, let $\rho:=\lim_{i\rightarrow\infty}\pi_i$.
This limit clearly exists: for any $i\mathbin\in\NN$ the first $2i$ transitions of $\rho$ are the first
$2i$ transitions of $\pi_i$ (and thus also of $\pi_j$ for any $j\mathbin>i$).
We show that $\rho$ is complete and embellishes~$\pi$.

For the latter property, the $i^{\rm th}$ transition $u_i$ of $\pi$ must also occur in $\pi_i$, and
no further than at position $2i$, for $\pi_i$ is an embellishment of $\pi$ obtained by adding only
$i$ transitions. As in the sequence $(\pi_j)_{j=0}^\infty$ past index $i$ no further changes occur
in the first $2i$ transitions, the transition $u_i$ also occurs in $\rho$.
Given the construction, this implies that $\rho$ embellishes $\pi$.
The same argument shows that $\rho$ embellishes $\pi_i$ for each $i\in\NN$.

Now suppose that $\rho$ is incomplete. Then there is a non-blocking transition $w$
that on $\rho$, from some position $k$ onwards, is continuously enabled.
Let $i\in\NN$ be an index such that $\pi_i$ is $(k,n)$-incomplete for some $k>h$.
Such an $i$ must exist, as the members of $(\pi_i)_{i=0}^\infty$ become less incomplete with
increasing $i$.
Let $M_h$ be the marking occurring at position~$h$ in $\rho$.
Then $M_h$ also occurs at position~$h$ in $\pi_i$, as the first $k+2i$ transitions of $\pi_i$
are the same for all $\pi_j$ with $j\geq i$, and thus for $\rho$.
Now $\precond{w}\leq M_h$ and \mbox{$\precond{w}\cap\precond{u}\mathbin=\emptyset$} for all transitions
$u$ occurring in $\rho$ past position $h$. Since $\rho$ embellishes $\pi_i$, this includes all transitions $u$ occurring
in $\pi_i$ past position $h$, so $w$ is continuously enabled from position $h$ onwards also on $\pi_i$,
contradicting the $(k,n)$-incompleteness of $\pi_i$.
\qed
\end{proof}

\begin{lemma}\label{lem:extension to complete path}
Any finite path in $N$ can be extended to a complete path in $N\!$, such that all transitions in the
extension have labels in $\OA\cup\{\tau\}$.
\end{lemma}
\begin{proof}
This is a simpler variant of the previous proof.
Let $\pi$ be the given path.
We build a sequence $\pi_i$ of paths in $N$ that all extend $\pi$, such that, for all $i$,
$\pi_{i+1}$ is less incomplete than $\pi_i$ and extends $\pi_i$ by one transition.

We start by taking $\pi_0$ to be $\pi$.
If at any point we hit a path $\pi_i$ that is complete, our work is done.
Otherwise, given the $(k,n)$-incomplete path $\pi_i$, for some $k$ and $n$,
pick a witness $w$ of the $(k,n)$-incompleteness of $\pi_i$
and obtain $\pi_{i+1} := \ph(\fs(\pi_i)w)$
by appending transition $w$ to $\pi_i$.
By construction $\pi_{i+1}$ extends $\pi_i$ by one non-blocking transition and hence extends $\pi$.
By \Lem{less incomplete} $\pi_{i+1}$ is less incomplete than $\pi_i$.

If at no point we hit a path $\pi_i$ that is complete, let $\rho:=\lim_{i\rightarrow\infty}\pi_i$.
Clearly, $\rho$ extends $\pi$. That $\rho$ is complete follows exactly as in the previous proof.
\qed
\end{proof}

\subsection{Paths of the Hypothetical Fair Scheduler}

\begin{lemma}\label{lem:PN1}
  Our hypothetical net $N$ has a path with no occurrences of (transitions labelled) $r_1$, but infinitely many occurrences
  of $r_2$.
\end{lemma}
\begin{proof}
We construct an infinite sequence $(\pi_k)_{k=0}^\infty$ of finite paths of $N$,
such that $\pi_k$ has no occurrences of $r_1$ and exactly $k$ occurrences
of $r_2$, and such that $\pi_k$ is a prefix of $\pi_{k+1}$ for all $k\in\NN$.
The limit of this sequence will be the required path.

$\pi_0$ is  the trivial path, consisting of the initial marking $M_0$ only.

Now assume we have constructed a path $\pi_k$ as required.
By \Lem{extension to complete path} $\pi_k$ can be extended into a complete path $\pi_k'$ that has no occurrences of $r_1$ and
exactly $k$ occurrences of $r_2$. Since $\pi_k'$ is complete, it must be $r_2$-enabled by Requirement~\ref{enabled PN}.
Hence there is a finite prefix $\pi_k''$ of $\pi_k'$, still extending $\pi_k$, such that a
transition~$v$ with $\ell(v)=r_2$ is enabled in the last state of $\pi_k''$.
Obtain $p_{k+1}$ by extending $\pi_k''$ with~$v$.
\qed
\end{proof}

\begin{lemma}\label{lem:PN2}
  $N$ has a path with exactly one occurrence of $r_1$, none of $t_1$, and infinitely many occurrences of $t_2$.
\end{lemma}
\begin{proof}
Let $\pi$ be the path found by \Lem{PN1}.
By \Lem{embellishment by complete path} this path is embellished by a complete path $\pi'$,
that thus has no occurrences of $r_1$ and infinitely many of $r_2$.
By Requirement~\ref{fair PN} $\pi'$ has infinitely many occurrences of $t_2$,
and by Requirement~\ref{patient PN} it has no occurrences of $t_1$.
By Requirement~\ref{enabled PN} $\pi'$ is $r_1$-enabled.
Let $w$ be a transition labelled $r_1$ that is on $\pi$ is continuously enabled from position $k$ onwards.
By \Lem{one step embellishment} $N$ has a path $\pi\oplus_k w$, obtained from $\pi'$ by inserting
transition $w$ in position $k$. That path has exactly one occurrence of $r_1$, none of $t_1$, and infinitely many of $t_2$.
\qed
\end{proof}

\begin{lemma}\label{lem:PN3}
  $N$ has a $t_1$-enabled path with infinitely many occurrences of $t_2$.
\end{lemma}

\begin{proof}
Let $\pi$ be the path found by \Lem{PN2}.
We build a sequence $\pi_i$ of paths in $N$ that all embellish $\pi$ and do not contain $t_1$, such that, for all $i$,
$\pi_{i+1}$ is less incomplete than $\pi_i$ and the first $2i$ transitions of $\pi_i$ and $\pi_{i+\!1}$ are the same.
Since the $\pi_i$ embellish $\pi$, they have exactly one occurrence of $r_1$, and infinitely many of $t_2$.
Moreover, by Requirement~\ref{fair PN}, none of the $\pi_i$ can be complete.

We start by taking $\pi_0$ to be $\pi$.
If at any point we hit a path $\pi_i$ that is $t_1$-enabled, our work is done.
Otherwise, given the $(k,n)$-incomplete path $\pi_i$, for some $k$ and $n$,
pick a witness $w$ of the $(k,n)$-incompleteness of $\pi_i$ and take $\pi_{i+1} \mathbin{:=} \pi_i\oplus_{k+2i}w$.
This path exists by \Lem{one step embellishment}, since $t$ is continuously enabled from position
$k$ onwards, and hence also from position $k+2i$ onwards.
Note that $\ell(w)\neq t_1$, since $\pi_i$ is not $t_1$-enabled.
Hence $\pi_{i+1}$ does not contain $t_1$.
By construction $\pi_{i+1}$ embellishes $\pi_i$ and hence $\pi$.
By \Lem{less incomplete} $\pi_{i+1}$ is less incomplete than $\pi_i$.
Moreover, the first $2i$ transitions of $\pi_i$ and $\pi_{i+\!1}$ are the same.

If at no point we hit a path $\pi_i$ that is $t_1$-enabled, let $\rho:=\lim_{i\rightarrow\infty}\pi_i$.
Exactly as in the proof of \Lem{embellishment by complete path} it follows that
$\rho$ is complete and embellishes $\pi$. Since $t_1$ does not occur on any of the $\pi_i$, it does
not occur on $\rho$. However, $r_1$ does occur on $\rho$, since it occurred on $\pi$.
This is in contradiction with Requirement~\ref{fair PN}.
Therefore, the assumption that at no point we hit a path $\pi_i$ that is $t_1$-enabled must be wrong.
\qed
\end{proof}

\begin{trivlist}
\item[\hspace{\labelsep}{\bf Proof of \Thm{no fair PN scheduler}}]
Let $\pi$ be the path found in \Lem{PN3}.
It must have a finite prefix $\pi'$ ending with an occurrence of $t_2$, such that a transition $w$
labelled $t_1$ is enabled it the last state of $\pi'$.
Extending $\pi'$ with $w$ yields a finite path of $N$ violating Requirement~\ref{coordinated PN}.
\qed
\end{trivlist}

\section{An Operational Petri Net Semantics of CCS}\label{sec:operational PN}

\newcommand{\weg}[1]{}
\newcommand{\src}{{\it src}}
\newcommand{\target}{{\it target}}
\newcommand{\shar}[1]{\mathord{\stackrel{#1}{\rightarrow}}}
\renewcommand{\myref}[1]{\hyperlink{lab:#1}{\sc (#1)}}

This section presents an operational Petri net semantics of CCS$^!\!$, following
Degano, De Nicola \mand Montanari \cite{DDM87}.
It associates a Petri net $\denote{P}$ with each CCS$^!$ expression $P$.
We establish that this Petri net is safe, all its reachable marking are
finite, and there are no transitions $u$ with $\precond{u}=\emptyset$;
hence it is one of the nets considered in \Sect{nets}. In \Sect{proof conclusion}
we will show that if a CCS$^!$ expression $F$ satisfies the four requirements of \Thm{no fair scheduler}
then the Petri net $\denote{F}$ satisfies the four requirements of \Thm{no fair PN scheduler}.
As a result, \Thm{no fair scheduler} will follow from \Thm{no fair PN scheduler}.

The standard operational semantics of CCS$^!$, presented in \Sect{CCS}, yields one big labelled
transition system for the entire language. Each individual CCS$^!$ expression $P$ appears as a state in
this LTS\@.  If desired, a \emph{process graph}---an LTS enriched with an initial state---for $P$
can be extracted from this system-wide LTS by appointing $P$ as the initial state, and optionally
deleting all states and transitions not reachable from $P$. In the same vein, an operational Petri
net semantics yields one big Petri net for the entire language, but without an initial marking.
We call such a Petri net {\em unmarked}. Each
process $P\in\T_{\rm CCS^!}$ corresponds to a marking $dec(P)$ of that net. If desired, a Petri net
for $P$ can be extracted from this system-wide net by appointing $dec(P)$ as its initial marking,
and optionally deleting all places and transitions not reachable from $dec(P)$.

The set $\SC_{\rm CCS^!}$ of places in the net---the \emph{grapes} of \cite{DDM87}---is the smallest set including:
\vspace{-12pt}
\begin{center}
\begin{tabular}{@{}l@{~~}l@{\qquad}l@{~~}l@{\qquad}l@{~~}l@{}}
&&$\AI$ &  \emph{agent identifier}\\
$\alpha.\E$  & \emph{prefixing}&
$\sum_{i\in I}\E_i$  & \emph{choice}&
$\mu\backslash a$  & \emph{restriction}\\
$\mu|$ & \emph{left parallel component}&
$|\mu$ & \emph{right component} &
$\mu[f]$ &  \emph{relabelling} \\
\end{tabular}
\end{center}
\noindent for $\AI\mathbin\in\K\!$, $\alpha\mathbin\in Act$,
 $\E,\E_i\mathbin\in\T_{\rm CCS^!}$, $a\mathbin\in\HC $, $\mu\mathbin\in\SC_{\rm CCS^!}$, index sets $I$, and relabellings~$f$.
The mapping $dec:\T_{\rm CCS^!} \rightarrow \pow(\SC_{\rm CCS^!})$ decomposing
a process expression into a set of grapes is inductively defined by:
\[
\begin{array}{@{}l@{~=~}l@{\qquad\qquad}l@{~=~}l@{}}
dec(\alpha.P)  & \{\alpha.P\} &
dec(A) & \{ A \} \\
dec(\sum_{i\in I}\E_i) & \{ \sum_{i\in I}\E_i \} &
dec(P|Q) & dec(P)|~ \cup ~|dec(Q) \\
dec(P\backslash a) & dec(P)\backslash a &
dec(P[f]) & dec(P)[f] \\
\end{array}
\]
Here $H[f]$, $H\backslash a$, $H|$ and $|H$ are understood element by element; e.g.\
$H[f] = \{\mu[f] \mid \mu\in H\}$. Moreover the binding is important, meaning that $(|H)|\not=|(H|)$.

We construct the unmarked Petri net $(S,T,F,\ell)$ of CCS$^!$
with $S:=\SC_{\rm CCS^!}$, specifying the triple $(T,F,\ell)$ as a ternary relation
$\mathord{\rightarrow} \subseteq \NN^S\times Act\times \NN^S$.
An element \plat{$H \ar{\alpha} J$} of this relation denotes a transition $u\mathbin\in\T$ with $\ell(u)\mathbin=\alpha$
such that $\precond{u}\mathbin=H$ and $\postcond{u}\mathbin=J$.
The transitions \plat{$H\ar{\alpha}J$} are derived from the rules of \Tab{PN-CCS}.
\begin{table}[ht]
\normalsize
\begin{center}
\framebox{$\begin{array}{@{}c@{}c@{}c@{}}
\{\alpha.P\} \ar{\alpha} dec(P)~~~~~~~~ &
\multicolumn{2}{c}{\displaystyle \frac{(dec(\E_j){-}K) \ar{\alpha} J}
{\{\sum_{i\in I}\E_i\} \ar{\alpha} J+K}~~(j\mathbin\in I,~~K\leq dec(P_j))}\\[4ex]
\displaystyle \frac{H \ar{\alpha} J}{H| \ar{\alpha} J|} &
\displaystyle \frac{H \ar{a} J \qquad K \ar{\bar a} L}{H| + |K  \ar{\tau} J| + |L} &
\displaystyle \frac{H \ar{\alpha} J}{|H \ar{\alpha} |J} \\[4ex]
\displaystyle\frac{H \ar{\alpha} J}{H\backslash a \ar{\alpha} J\backslash a}~
                     (a\mathbin{\neq}\alpha\mathbin{\neq}\bar{a}) &
\displaystyle\frac{H \ar{\alpha} J}{H[f] \ar{f(\alpha)} J[f]} &
~~\displaystyle \frac{(dec(\E){-}K) \ar{\alpha} J}{\{\AI\} \ar{\alpha} J+K}~\!%
                     \Big(\begin{array}{@{}c@{}}\scriptstyle A \mathbin{\stackrel{{\it def}}{=}} P\\[-3pt]
                     \scriptstyle K\mathord\leq \textit{dec}(P) \end{array}\Big)\!
\end{array}$}
\end{center}
\caption{Operational Petri net semantics of CCS$^!$}
\label{tab:PN-CCS}
\end{table}

Henceforth, we write $M\mathrel{[\alpha\rangle} M'$, for markings $M,M'\in \NN^{\SC_{\rm CCS^!}}$ and
$\alpha\mathbin\in\Act$, if there exists a transition $u$ with $M[u\rangle M'$ and $\ell(u)=\alpha$.
In that case \plat{$M=H+K$} and \plat{$M'=J+K$} for multisets of places $H,J,K: \SC_{\rm CCS^!}\rightarrow\NN$ with \plat{$H\ar{\alpha}J$}.

The following theorem says that function $dec$ is a strong bisimulation (\cite{Mi89}) between the LTS and the unmarked
Petri net of CCS$^!\hspace{-2.5pt}$. Since markings of the form $dec(R)$ are plain sets (rather than multisets), it also
follows that the Petri net of each CCS$^!$ expression is safe.%
\begin{theorem}\rm\label{thm:bisimulation}
If $R \ar{\alpha} R'$ for $R,R'\mathbin\in\T_{\rm CCS^!}$ and $\alpha\mathbin\in\Act$ then $dec(R) \mathrel{[\alpha\rangle} dec(R')$.
Moreover, if $dec(R) \mathbin{[\alpha\rangle} M$ then there is a $R'\mathbin\in\T_{\rm CCS^!}$ with $R \ar{\alpha} R'$ and $dec(R')\mathbin=M$.
\end{theorem}
\begin{proof}
The first statement follows by induction on the derivability of the transition
\plat{$R \ar{\alpha} R'$} from the rules of \Tab{CCS}. We only spell out two representative cases;
the others are similar or straightforward.
\begin{itemize}
\item Suppose $P|Q \ar{\alpha} P'|Q$ because $P \ar{\alpha} P'$.
By induction $dec(P) \mathrel{[\alpha\rangle} dec(P')$.
Hence \plat{$dec(P)=H+ K$} and \plat{$dec(P')=J+K$} for (multi)sets $H,J,K\subseteq \SC_{\rm CCS^!}$ with \plat{$H\ar{\alpha}J$}.
By \Tab{PN-CCS} we obtain \plat{$H| \ar{\alpha} J|$}. Hence
$$\begin{array}{ccl}dec(P|Q) & = & dec(P)| + |dec(Q) \\
  & = & H| + K| + |dec(Q) \\
  & \mathrel{[\alpha\rangle} & J| + K| + |dec(Q) \\
  & = & dec(P')| + |dec(Q) \\
  & = & dec(P'|Q)\;.
\end{array}$$
\item Suppose $A \mathbin{\stackrel{{\it def}}{=}} P$ and $A \ar{\alpha} P'$ since $P \ar{\alpha} P'$.
By induction $dec(P) \mathrel{[\alpha\rangle} dec(P')$. Hence  \plat{$dec(P)=H+K$} and
\plat{$dec(P')=J+ K$} for sets $H,J,K\subseteq \SC_{\rm CCS^!}$ with \plat{$dec(P){-}K=H\ar{\alpha}J$}.
By \Tab{PN-CCS}, \plat{$dec(A) = \{A\} \ar{\alpha} J + K = dec(P')$}.
\end{itemize}
The second statement can be reformulated as
\begin{quote}
if $(dec(R){-}K) \mathbin{\ar{\alpha}} J$ with $K\leq dec(R)$\\ then there is a $R'\mathbin\in\T_{\rm CCS^!}$ with $R \ar{\alpha} R'$ and $dec(R')\mathbin=J{+}K$.
\end{quote}
for $R\mathbin\in\T_{\rm CCS^!}$ and $K,J: \SC_{\rm CCS^!}\rightarrow\NN$.
We prove it by induction on the derivability of the transition \plat{$dec(P){-}K \ar{\alpha} J$} from the rules of \Tab{PN-CCS}.
\begin{itemize}
\item Suppose $dec(R)-K = \{\alpha.P\} \ar{\alpha} dec(P)$. Since the only set $dec(R)$ containing
  $\alpha.P$ is $\{\alpha.P\}$, we have $K\mathbin=\emptyset$, $J\mathbin=dec(P)$ and $R\mathbin=\alpha.P$.  Take $R':=P$.
\item Suppose $dec(R)-K' =H[f] \ar{f(\alpha)} J[f]$ because $H \ar{\alpha} J$.
  Then $R$ must have the form $P[f]$, so that $dec(R)=dec(P)[f]$, and $K'$ must have the form $K[f]$.\linebreak[4]
  Thus \plat{$dec(P)-K=H\ar{\alpha} J$}, and by induction there is a $P'\mathbin\in\T_{\rm CCS^!}$ with
  \plat{$P \ar{\alpha} P'$} and $dec(P')\mathbin=J+K$. By \Tab{CCS}, \plat{$R = P[f] \ar{\alpha} P'[f]$}.
  Moreover, $dec(P'[f])=dec(P')[f]=J[f]+K[f]=J[f]+K'$.
\item The case for restriction proceeds likewise.
\item Suppose $dec(R)-K' = H| \ar{\alpha} J|$ because $H \ar{\alpha} J$.
  Then $R$ must have the form $P|Q$, and $K' = (dec(P)| - H|) + |dec(Q) = K| + |dec(Q)$, where $K:=dec(P)-H$.
  Thus \plat{$dec(P)-K=H\ar{\alpha} J$}, so by induction there is a $P'\mathbin\in\T_{\rm CCS^!}$ with
  \plat{$P \ar{\alpha} P'$} and $dec(P')\mathbin=J+K$.  By \Tab{CCS}, \plat{$R = P|Q \ar{\alpha} P'|Q$}.
  Moreover, $dec(P'|Q)=dec(P')|+|dec(Q)=J|+K|+|dec(Q)=J|+K'$.
\item Suppose $dec(R)-K' = H| + |K  \ar{\tau} J| + |L$ because $H \ar{a} J$ and $K \ar{\bar a} L$.
  Then $R$ has the form $P|Q$, and $K' = (dec(P)| - H|) + (|dec(Q) - |J)= K_1| + |K_2$, where
  $K_1:=dec(P)-H$ and $K_2:=dec(Q)-J$.
  Thus \plat{$dec(P)-K_1=H\ar{a} J$} and \plat{$dec(Q)-K_2=J\ar{\bar a} L$}, so by induction
  there are $P',Q'\mathbin\in\T_{\rm CCS^!}$ with \plat{$P \ar{a} P'$}, $dec(P')\mathbin=J+K_1$,
  \plat{$Q \ar{\bar a} Q'$} and $dec(Q')\mathbin=L+K_2$.  By \Tab{CCS}, \plat{$R = P|Q \ar{\tau} P'|Q'$}.
  Moreover, $dec(P'|Q')=dec(P')|+|dec(Q')=J|+K_1|+|L+|K_2=J|+|L+K'$.
\item The case for the last rule for parallel composition follows by symmetry.
\item Suppose $dec(R)-K' = \{\sum_{i\in I}\E_i\} \ar{\alpha} J+K$ because $(dec(\E_j){-}K) \ar{\alpha} J$
  for some $j\mathbin\in I$.
  Since the only set $dec(R)$ containing $\sum_{i\in I}\E_i$ is $\{\sum_{i\in I}\E_i\}$,
  we have $K'\mathbin=\emptyset$ and $R=\sum_{i\in I}\E_i$.
  By induction, there is a $P'_j\mathbin\in\T_{\rm CCS^!}$ with \plat{$P_j \ar{\alpha} P'_j$}
  and $dec(P'_j)\mathbin=J{+}K$.
  By \Tab{CCS}, \plat{$R = \sum_{i\in I}\E_i \ar{\alpha} P'_j$}.
\item The case
  for recursion (agent identifiers) goes likewise.
\qed
\end{itemize}
\end{proof}

A trivial induction shows that there are no transitions without preplaces.
The following lemma implies that all reachable markings are finite, so
that the Petri nets of CCS$^!$ expressions
have all the properties of nets imposed in \Sect{nets}.
\begin{lemma}
For any $P\mathbin\in\T_{\rm CCS^!}$ the set $dec(P)$ is finite.
\end{lemma}
\begin{proof}
A straightforward induction.
\qed\end{proof}

The above operational Petri net semantics of CCS has the disadvantage that initial concurrency
in expressions of the form $\sum_{i\in I}\E_i$ or $A$ is not represented~\cite{DDM87}.
Although this Petri net semantics matches the LTS semantics of CCS up
to strong (interleaving) bisimilarity---and thereby also the standard denotational
Petri net semantics of CCS-like operators \cite{GV87}---, it
does not match the standard denotational Petri net semantics up to semantic equivalences that take
concurrency explicitly into account. For this reason Olderog~\cite{Old91} provides an alternative
operational Petri net semantics that is more accurate in this sense. However, the work of Olderog
does not generalise in a straightforward way to the infinite sum construct of CCS, and to unguarded
recursion. 
In fact, a safe Petri net that accurately models the concurrent behaviour of the CCS process
$\sum_{i\in \NN}(a_i.0|b_i.0)$ would need an uncountable initial marking, and hence falls outside
the class of nets we handle in \Sect{nets}. Since the accurate modelling of concurrency is not
essential for this paper, we therefore use the semantics of \cite{DDM87}.

\section[Fair Schedulers Cannot be Rendered in CCS---Proof]
        {Fair Schedulers Cannot be Rendered in CCS$^{!}$---Proof}
\label{sec:proof conclusion}

\begin{lemma}\label{lem:injective}
The mapping $dec:\T_{\rm CCS^!} \rightarrow \pow(\SC_{\rm CCS^!})$ is injective.
\end{lemma}
\begin{proof}
A straightforward induction on the structure of the elements in $\T_{\rm CCS^!}$.
\qed
\end{proof}

\begin{lemma}\label{lem:path transfer}
Let $P\mathbin\in\T_{\rm CCS^!}$.
For any path $\pi=dec(P) u_1 M_1 u_2 M_2 u_3 \dots$ in the unmarked Petri net of CCS$^!$
there is a unique path \plat{$\widehat\pi = P \ar{\alpha_1} P_1 \ar{\alpha_2} P_2 \ar{\alpha_3}
\dots$} of the same (finite or infinite) length in the LTS of CCS$^!$ with $dec(P_i)\mathbin=M_i$ and $\ell(u_i)\mathbin=\alpha_i$ for
all $i$.
\end{lemma}
\begin{proof}
A straightforward induction on $i$, using \Thm{bisimulation} and \Lem{injective}.
\qed
\end{proof}

\noindent
The following observations are based directly on \Tab{PN-CCS} and the definition of $dec$.%
\begin{observation}\label{obs:decompose parallel}\rm Let $dec(P|Q)\mathrel{[u\rangle} M$.
Then $M$ has the form $dec(P'|Q')$ and either
\begin{itemize}
\vspace{-1ex}
\item $Q\mathbin=Q'$ and $dec(P)\mathrel{[v\rangle} dec(P')$ for a $v\mathbin\in T$ with $\ell(v)\mathbin=\ell(u)$,
  $\precond{u}\mathbin=\precond{v}|$ and $\postcond{u}\mathbin=\postcond{v}|$,
\item $P\mathbin=P'$ and $dec(Q)\mathrel{[w\rangle} dec(Q')$ for a $w\mathbin\in T$ with $\ell(w)\mathord=\ell(u)$,
  $\!\precond{u}\mathord=|\precond{w}$ and $\postcond{u}\mathord=|\postcond{w}\!\!$,
\item or $dec(P)\mathrel{[v\rangle} dec(P')$ and $dec(Q)\mathrel{[w\rangle} dec(Q')$ for
  $v,w\mathbin\in T$ with $\ell(v)\mathbin=c\in\HC$, $\ell(w)\mathbin=\bar c$,
  $\precond{u}\mathbin=\precond{v}|+|\precond{w}$ and $\postcond{u}\mathbin=\postcond{v}|+|\postcond{w}$.
\end{itemize}
\end{observation}
For each such transition $u$, the transitions $v$ and $w$ discovered above are called the left- and
right-projections of $u$, respectively, when they exist.
Hence any path $\pi$ starting from a marking $dec(P|Q)$ can be uniquely decomposed into a 
path $\pi_1$ starting from $dec(P)$ and a path $\pi_2$ starting from $dec(Q)$, notation $\pi\Rrightarrow\pi_1|\pi_2$. 
The path $\pi_1$ fires all existing left-projections of the transitions in $\pi$, in order, and
$\pi_2$ its right-projections.

\begin{lemma}\label{lem:enabled path}
If $\pi \Rrightarrow \pi_1 | \pi_2$ and $\pi_1$ is $\alpha$-enabled, then so is $\pi$. 
\end{lemma}
\begin{proof}
Let $\pi_1 = M_0' v_1 M_1' v_2 M_2' v_3 \dots$ be $\alpha$-enabled. Then there is a $h\geq 0$ and a transition $v$ with
$\ell(v)=\alpha$ such that $M_h' [v\rangle$ and $\precond{v}\cap\precond{v_i}=\emptyset$ for all $i \mathbin > h$.
Let $\pi = M_0 u_1 M_1 u_2 M_2 u_3 \dots$. 
Let $k\geq h$ be such that $v_0 \dots v_h$ is the sequence of existing left-projections of $u_1\dots u_k$.
The marking $M_k$ has the form $dec(P_k|Q_k)=dec(P_k)| \dcup |dec(Q_k)$ with $M_h'=dec(P_k)$.
Since $\precond{v}\leq dec(P_k)$, by \Tab{PN-CCS} there exists a transition $v'$ with
$\ell(v')=\alpha$ and $\precond{v'}=\precond{v}| \leq dec(P_k)| \leq M_k$. So $M_k[v'\rangle$.
Moreover, since $\precond{v}\cap\precond{v_i}=\emptyset$ for all $i>h$ we have $\precond{v'}\cap
\precond{u_i}=\emptyset$ for all $i>k$.
It follows that $\pi$ is $\alpha$-enabled.
\qed
\end{proof}

\begin{lemma}\label{lem:enabled path 2}
If $\pi \Rrightarrow \pi_1 | \pi_2$, $\pi_1$ is $c$-enabled and $\pi_2$ is
$\bar c$-enabled, for some $c\in\HC$, then $\pi$ is $\tau$-enabled.
\end{lemma}
\begin{proof}
Let $\pi = M_0 u_1 M_1 u_2 M_2 u_3 \dots$. Since $ \pi \Rrightarrow  \pi_1 |  \pi_2$,
each marking $M_i$ has the form $dec(P_i|Q_i) = dec(P_i)| \dcup |dec(Q_i)$.
By the same reasoning as in the previous proof, there is a $k_1\geq 0$ and a transition $v$ with $\ell(v)=c$ such that
$\precond{v} \leq dec(P_{i})$ and $\precond{v}|\cap \precond{u_i}=\emptyset$ for all $i\geq k_1$.
Likewise, there is a $k_2\geq 0$ and a transition $w$ with $\ell(w)=\bar c$ such that
$\precond{w} \leq dec(Q_{i}|$ and $|\precond{w}\cap \precond{u_i}=\emptyset$ for all $i\geq k_2$.

Let $k=\max(k_1,k_2)$. By the fourth rule of \Tab{PN-CCS} there is a transition $u$ with
$\ell(u)\mathbin=\tau$ and $\precond{u}\mathbin=\precond{v}|+|\precond{w} \mathbin\leq M_k$ and
$\precond{u}\cap \precond{u_i}\mathbin=\emptyset$ for all $i> k$. So $\pi$ is $\tau$-enabled.
\qed
\end{proof}

\begin{observation}\label{obs:decompose restriction}\rm
Let $dec(P\backslash c)\mathrel{[u\rangle} M$.
Then $M$ has the form $dec(P'\backslash c)$ and we have
$dec(P)\mathrel{[v\rangle} dec(P')$ for a $v\mathbin\in T$ with $\ell(v)\mathbin=\ell(u)$,
$\precond{u}\mathbin=\precond{v}\backslash c$ and $\postcond{u}\mathbin=\postcond{v}\backslash c$.
\end{observation}
For each such transition $u$, the transition $v$ discovered above is called the projection of $u$.
Hence any path $\pi$ starting from a marking $dec(P)\backslash c$ can be uniquely decomposed into a 
path $\pi'$ starting from $dec(P)$, notation $\pi\Rrightarrow\pi'\backslash c$. 
The path $\pi'$ fires all the projections of the transitions in $\pi$, in order.

\begin{lemma}\label{lem:enabled path 3}
If $ \pi \Rrightarrow  \pi'\backslash c$ and $\pi'$ is $\alpha$-enabled with $c \neq \alpha
\neq \bar c$, then $\pi$ is $\alpha$-enabled.
\end{lemma}
\begin{proof}
Let $\pi' = M_0' v_1 M_1' v_2 M_2' v_3 \dots$ be $\alpha$-enabled. Then there is a $k\geq 0$ and a transition $v$ with
$\ell(v)=\alpha$ such that $M_k' [v\rangle$ and $\precond{v}\cap\precond{v_i}=\emptyset$ for all $i \mathbin > k$.
Let $\pi = M_0 u_1 M_1 u_2 M_2 u_3 \dots$. 
The marking $M_k$ has the form $dec(P_k\backslash c)=dec(P_k)\backslash c$ with $M_k'=dec(P_k)$.
Since $\precond{v}\leq dec(P_k)$ and $c\neq\alpha\neq\bar{c}$, by \Tab{PN-CCS} there exists a transition $v'$ with
$\ell(v')=\alpha$ and $\precond{v'}=\precond{v}\backslash c \leq dec(P_k)\backslash c = M_k$. So $M_k[v'\rangle$.
Moreover, since $\precond{v}\cap\precond{v_i}=\emptyset$ for all $i>k$ we have $\precond{v'}\cap
\precond{u_i}=\emptyset$ for all $i>k$. It follows that $\pi$ is $\alpha$-enabled.
\qed
\end{proof}

\begin{observation}\label{obs:decompose relabelling}\rm
Let $dec(P[f])\mathrel{[u\rangle} M$.
Then $M$ has the form $dec(P'[f])$ and we have
$dec(P)\mathrel{[v\rangle} dec(P')$ for a $v\mathbin\in T$ with $f(\ell(v))\mathbin=\ell(u)$,
$\precond{u}\mathbin=\precond{v}[f]$ and $\postcond{u}\mathbin=\postcond{v}[f]$.
\end{observation}
For each such transition $u$, the transition $v$ discovered above is called a projection of $u$---it
need not be unique.
Hence any path $\pi$ starting from a marking $dec(P)\backslash c$ can be decomposed into a 
path $\pi'$ starting from $dec(P)$, notation $\pi\in \pi'[f]$. 
The path $\pi'$ fires projections of the transitions in $\pi$, in order.

\begin{lemma}\label{lem:enabled path 4}
If $ \pi \in  \pi'[f]$ and $\pi'$ is $\alpha$-enabled, then $\pi$ is $f(\alpha)$-enabled.
\end{lemma}
\begin{proof}
Just as the proof of \Lem{enabled path 3}.
\qed
\end{proof}

\begin{observation}\label{obs:decompose projection}\rm
Let $\pi$ be a path in the unmarked Petri net of CCS$^!$.\\
Then $\pi\Rrightarrow\pi_1|\pi_2$ implies 
{that  $\widehat\pi_1|\widehat\pi_2$ is a decomposition of $\widehat\pi$ (c.f.\ Page~\hyperlink{hr:decomp}{\pageref*{pg:decomp}})}.
\\
Likewise, if $\pi\Rrightarrow\pi'\backslash c$ or $\pi\in\pi'[f]$ then $\widehat\pi'$ is a
decomposition of $\widehat\pi$.
\end{observation}

\begin{proposition}\label{prop:enabled just}
Let $\pi$ be a path in the unmarked Petri net of CCS$^!$.
If $Y$ includes all actions $\alpha\in\Act$ for which $\pi$ is $\alpha$-enabled, and $Y\subseteq\HC$,
then $\widehat\pi$ is $Y\!$-just.
\end{proposition}

\newcommand{\ejustn}{just${}_{en}$}
\newcommand{\ejust}{just${}_{en}$\ }
\newcommand{\Obs}[1]{Observation~\ref{obs:#1}}
\begin{proof}
  Define a path $\eta$ in the LTS of CCS$^!$ to be $Y\!$-\emph{\ejustn}, for $Y\subseteq\Act$,
  if $\eta$ has the form $\widehat\pi$ for $\pi$ a path in the unmarked Petri net of CCS$^!$, and
  $Y$ includes all actions $\alpha\mathbin\in\Act$ for which $\pi$ is $\alpha$-enabled.
  Note that if $\eta$ is $Y\!$-\ejustn, it is also $Y'\!$-\ejust for any $Y\subseteq Y'\subseteq\Act$.
  We show that the family of predicates $Y\!$-justness${}_{en}$, for $Y\mathbin\subseteq \HC\!$,
  satisfies the five requirements of \Def{just path}.
  
  \begin{itemize}
  \item Let $\widehat \pi$ be a finite $Y\!$-\ejust path. Suppose the last state $Q$ of $\widehat\pi$ admits an action $\alpha\notin Y$,
  Then, using \Thm{bisimulation}, the last marking $dec(Q)$ of $\pi$ enables a
  transition labelled $\alpha$. Thus $\pi$ is $\alpha$-enabled, contradicting the $Y\!$-justness${}_{en}$ of $\pi$.

  \item Suppose $\widehat\pi$ is a $Y\!$-\ejust path of a process $P|Q$ {with $Y\mathbin\subseteq \HC\!$}.
  Then $Y$ includes all actions $\alpha\mathbin\in\Act$ for which $\pi$ is $\alpha$-enabled.
  Let $\pi_1$ and $\pi_2$ be the paths such that $\pi\Rrightarrow \pi_1|\pi_2$.
  By \Obs{decompose projection} $\widehat \pi$ can be decomposed into the paths $\widehat\pi_1$ of $P$ and $\widehat\pi_2$ of $Q$.
  Let $X \subseteq Act$ be the set of actions $\alpha$ for which $\pi_1$ is $\alpha$-enabled,
  and let $Z \subseteq Act$ be the set of actions $\alpha$ for which $\pi_2$ is $\alpha$-enabled.
  By definition, $\widehat\pi_1$ is $X$-\ejust and $\widehat\pi_2$ is $Z$-\ejustn. 
  
  If $\pi_1$ is $\alpha$-enabled then $\pi$ is $\alpha$-enabled by \Lem{enabled path}.
  This implies that $X\subseteq Y$.
  In the same way it follows that $Z\subseteq Y$.
  Now suppose $X\cap \bar{Z} \neq\emptyset$.
  Then $\pi_1$ is $c$-enabled and $\pi_2$ is $\bar c$-enabled, for some $c\in\HC$.
  So, by \Lem{enabled path 2}, $\pi$ is $\tau$-enabled, in contradiction with $\tau\not\in Y\subseteq \HC$.
  We therefore conclude that $X\cap \bar{Z} = \emptyset$.

  \item Suppose $\widehat\pi$ is a $Y\!$-\ejust path of a process $P\backslash c$.
  Then $Y$ includes all actions $\alpha\mathbin\in\Act$ for which $\pi$ is $\alpha$-enabled.
  Let $\pi'$ be the path such that $\pi\Rrightarrow \pi'\backslash c$.
  By \Obs{decompose projection} $\widehat \pi$ is a decomposition of the path $\widehat\pi'$ of $P$.
  Let $X \subseteq Act$ be the set of actions $\alpha$ for which $\pi'$ is $\alpha$-enabled.
  If $\pi'$ is $\alpha$-enabled with $c\neq\alpha\neq \bar c$ then $\pi$ is $\alpha$-enabled by \Lem{enabled path 3}.
  This implies that $X\setminus\{c,\bar c\} \subseteq Y$ and hence $X\subseteq \HC$.
  It follows that $\widehat\pi'$ is $X$-\ejustn, and hence $Y\mathord\cup\{c,\bar c\}$-\ejustn.

  \item Suppose $\widehat\pi$ is a $Y\!$-\ejust path of a process $P[f]$.
  Then $Y$ includes all actions $\alpha\mathbin\in\Act$ for which $\pi$ is $\alpha$-enabled.
  Let $\pi'$ be a path such that $\pi\in \pi'[f]$.
  By \Obs{decompose projection} $\widehat \pi$ is a decomposition of the path $\widehat\pi'$ of $P$.
  Let $X \subseteq Act$ be the set of actions $\alpha$ for which $\pi'$ is $\alpha$-enabled.
  If $\pi'$ is $\alpha$-enabled then $\pi$ is $f(\alpha)$-enabled by \Lem{enabled path 4}.
  This implies that $f(X) \subseteq Y$.
  It follows that $\widehat\pi'$ is $X$-\ejustn, and hence $f^{-1}(Y)$-\ejustn.

  \item Suppose $\pi'$ is a suffix of an $Y\!$-\ejust path $\pi$.
  Then $Y$ includes all actions $\alpha\mathbin\in\Act$ for which $\pi$ is $\alpha$-enabled, and
  thus all $\alpha$ for which $\pi'$ is $\alpha$-enabled. Hence $\pi'$ is $Y\!$-\ejustn.
  \end{itemize}
  Since $Y\!$-justness is the largest family of
  predicates that satisfies those requirements, $Y\!$-justness${}_{en}$ implies $Y\!$-justness.
\qed\end{proof}

\begin{corollary}\label{cor:enabling transfer}
Let $\pi$ be a path starting from $dec(P)$ in the unmarked Petri net of CCS$^!$.
If $\pi$ is complete, then so is $\widehat\pi$.
Moreover, if $\widehat\pi$ is $a$-enabled, for $a\mathbin\in\HC$, then so is $\pi$.
\qed
\end{corollary}

\noindent
\textbf{Proof of \Thm{no fair scheduler}}~
Suppose there \emph{does} exist a CCS$^!$ expression $F$ as considered in \Thm{no fair scheduler}.
Then it suffices to show that $\denote{F}$ is a net $N$ as considered in \Thm{no fair PN scheduler}.
Thus, we show that $\denote{F}$ satisfies the four properties of \Thm{no fair PN scheduler}.
\begin{enumerate}
\item Let $\pi$ be a complete path of $\denote{F}$ that has finitely many occurrences of $r_i$.
  By \Lem{path transfer} $\widehat\pi$ is a path of $F$ that has finitely many occurrences of $r_i$.
  By \Cor{enabling transfer} it is complete. By Requirement 1 of \Thm{no fair scheduler}, $\widehat\pi$
  is $r_i$-enabled. So by \Cor{enabling transfer}, $\pi$ is $r_i$-enabled.
\item Let $\pi$ be a complete path of $\denote{F}$. Then $\widehat\pi$ is a complete path of $\denote{F}$.
  By Requirement 2 of \Thm{no fair scheduler}, on $\widehat\pi$ each $r_i$ is followed by a $t_i$.
  Using \Lem{path transfer}, the same holds for $\pi$.
\item Let $\pi$ be a finite path of $\denote{F}$.
  Then $\widehat\pi$ is a path of $F$.
  By Requirement 3 of \Thm{no fair scheduler}, on $\widehat\pi$, and thus on $\pi$, are no
  more occurrences of $t_i$ than of $r_i$.
\item Let $\pi$ be a path of $\denote{F}$, featuring two occurrences of $t_i$ and $t_j$ ($i,j\in\{1,2\}$).
  These occurrences also occur on $\widehat\pi$.
  By Requirement 4 of \Thm{no fair scheduler}, an action $\e$ occurs between them.
\qed
\end{enumerate}

\section{Concluding Remarks}\label{sec:conclusion} 
\newcommand{\prio}{\mathbin{\rhd}}
This paper presented a simple fair scheduler---one that in suitable variations occurs in many
distributed systems---of which no implementation can be expressed in CCS\@.
In particular, Dekker's and Peterson's mutual exclusion protocols cannot be rendered correctly in CCS.
These conclusions remain true if CCS is extended with progress and certain fairness assumptions, namely justness as presented in this paper.
However, as shown in~\cite{CorradiniEtAl09}, it is possible to correctly render Dekker's
protocol---and thereby a fair scheduler---in CCS enriched with a stronger fairness assumption. 
We argue, however, that such fairness assumptions should not be made lightly, as in certain contexts they
allow the derivation of false results.

It does not appear hard to extend CCS with an operator that enables expressing this fair
scheduler without relying on a fairness assumption.
In \cite{GH14} for instance we give a simple specification of a fair scheduler in an extension of
CCS with broadcast communication. In \cite{CDV09} it is shown that it suffices (for the correct specification
of Dekker's algorithm) to extend a CCS-like process algebra with non-blocking reading actions.
A priority mechanism \cite{CLN01} would also be sufficient.

Let $\prio$ for instance be a +-like operator that schedules an action from its left argument if
possible, and otherwise runs its right argument. Then $\widehat F_1$, with
\[
F_1\stackrel{\it def}{=} (r_1.t_1.e.F_2) \prio (\tau.F_2)\quad\mbox{and}\quad
F_2\stackrel{\it def}{=} (r_2.t_2.e.F_1) \prio (\tau.F_1)
\]
appears to be a fair scheduler. Here
$\widehat\cdot$ is the CCS-context specified in \Sect{characterisation}.

$F$ is basically a round-robin scheduler which checks whether $r_1$ is enabled; if so, it performs
the sequence $r_1.t_1.e$; if not, it does an internal action and tries to perform $r_2$.

An interesting question is what kind of extension of CCS is needed to enable specifying all
processes of this kind. It appears that the formalism CCS+LTL that we employed in \Sect{fair spec} to
specify our fair scheduler can be used to specify a wide range of similar systems.
Such a specification combines a CCS specification with a fairness component, consisting
of a set of LTL formulas that narrows down the complete paths of the specified process.
An intriguing challenge is to find an extension of CCS, say by means of extra operators,
that makes the fairness component redundant, i.e.\ an extension such that any CCS+LTL process can be expressed in
the extended CCS without employing a fairness component.

For certain properties of the form $(\bigvee_{i}\mathbf{G}\mathbf{F}a_{i}) \Rightarrow (\bigvee_{j}\mathbf{G}\mathbf{F}b_{j})$
where the $a_i$ and $b_i$ are action occurrences%
---hence for specific strong fairness properties---one can
define a \emph{fairness operator} that
transforms a given LTS into a LTS that satisfies the property~\cite{PuhakkaValmari01}.
This is done by eliminating all the paths that do not satisfy the property via a carefully designed parallel composition. 
The fairness operator can be expressed in a variant of the process algebra CSP\@.
The question above asks whether something similar can be done, in a more expressive process algebra,
for arbitrary LTL properties, or perhaps for a larger class of fairness properties.

 \begin{acknowledgements}
We gratefully thank the anonymous referees. Their reports showed deep insights in the
material, and helped a lot to improve the quality of the paper. In particular, the link between our fair
scheduler and Peterson's mutual exclusion protocol was made by one of the referees.
 \end{acknowledgements}

\newcommand{\Ruediger}{-R.}
\def\SSort#1{}\def\NSort#1{}

\end{document}